\documentclass[12pt]{article}
\usepackage{amssymb,amsmath,amsfonts,amsthm,geometry,ulem,graphicx,caption,xcolor,setspace,sectsty,comment,footmisc,natbib,pdflscape,subfigure,array,multicol,multirow,float,tikz,appendix}
\usepackage[hyperindex,breaklinks]{hyperref}
\usepackage{cleveref}

\newtheorem{theorem}{Theorem}
\newtheorem{corollary}{Corollary}
\newtheorem{proposition}{Proposition}
\newtheorem{lemma}{Lemma}

\theoremstyle{definition}
\newtheorem{definition}{Definition}
\newtheorem{example}{Example}
\newtheorem{remark}{Remark}

\newcolumntype{L}[1]{>{\raggedright\let\newline\\arraybackslash\hspace{0pt}}m{#1}}
\newcolumntype{C}[1]{>{\centering\let\newline\\arraybackslash\hspace{0pt}}m{#1}}
\newcolumntype{R}[1]{>{\raggedleft\let\newline\\arraybackslash\hspace{0pt}}m{#1}}

\tikzset{global scale/.style={
    scale=#1,
    every node/.append style={scale=#1}}
}
\hypersetup{colorlinks = true, linkcolor = black, urlcolor = black, citecolor = black}
\newcommand\posscite[1]{\citeauthor{#1}'s\ (\citeyear{#1})}
\begin{document}
\title{Games under the Tiered Deferred Acceptance Mechanism\thanks{I am especially grateful to Federico Echenique, Haluk Ergin, Yuichiro Kamada, and Xiaohan Zhong for their guidance and support. I thank Itai Ashlagi, Yan Chen, Peter Doe, Zhiming Feng, Sam Kapon, Fuhito Kojima, Nozomu Muto, Axel Niemeyer, Juan Sebastián Pereyra, Luciano Pomatto, Al Roth, Kirill Rudov, Chris Shannon, Pan-Yang Su, Takuo Sugaya, Omer Tamuz, Qianfeng Tang, Alex Teytelboym, Quitzé Valenzuela-Stookey, Xingye Wu, Ning Yu, Jun Zhang, and the seminar participants at Caltech, UC Berkeley, Stony Brook Game Theory Conference 2024, GAMES 2024 for their helpful comments and insights. A part of this research was conducted while I was visiting the Department of Economics at UC Berkeley.}}
\author{Jiarui Xie\thanks{School of Economics and Management, Tsinghua University. Email: \href{xiejiarui21@gmail.com}{xiejiarui21@gmail.com}.}}

\date{\today}

\maketitle

\begin{abstract}
We study the tiered deferred acceptance mechanism used in school admissions, such as in China and Turkey. 
This mechanism partitions schools into tiers and applies the deferred acceptance algorithm within each tier. Once assigned, students cannot apply to schools in subsequent tiers. 
We show that this mechanism is not strategy-proof. 
In the induced preference revelation game, we find that merging tiers preserves all equilibrium outcomes, and within-tier acyclicity is necessary and sufficient for the mechanism to implement stable matchings. 
We also find that introducing tiers to the deferred acceptance mechanism may not improve student quality at top-tier schools as intended.
\\
\vspace{0in}
\vspace{0in}
\end{abstract}

\newpage
\section{Introduction}
Many matching mechanisms allocate resources sequentially, with some resources allocated first and the rest later. We study one such mechanism, known as the tiered deferred acceptance (TDA) mechanism, within the context of school choice. In this mechanism, schools are tiered, and seats are allocated sequentially by tier. Within each tier, the student-proposing deferred acceptance (DA) algorithm is used to assign seats. Students can only apply to schools within the current tier, and once assigned at the end of one tier, they cannot apply to schools in subsequent tiers.\footnote{In the DA mechanism, students first apply to their favorite schools, which then tentatively accept or reject students based on their priorities. This process repeats, with rejected students applying to their next choices, until no more students are rejected. See \Cref{sec:mech} for the formal definitions of the DA and TDA mechanisms.}

The TDA mechanism is used in various settings. In the Chinese college admission system, schools are partitioned into three tiers (in the following order): specialized institutions (e.g., military or diplomatic schools), academic universities, and vocational education colleges. Some specialized institutions and vocational colleges rank students using the national college entrance exam and school-specific tests, while other schools use only the national exam score.\footnote{\label{ft:China}In some provinces, academic universities are further partitioned into two tiers: key universities first and other universities second. In 2023, 14 provinces, including Beijing, Shanghai, and Jiangsu, did not have this extra partition; while 16 provinces, including Shanxi and Sichuan, did; only Inner Mongolia used a dynamic adjustment system. Data source: \href{https://gaokao.chsi.com.cn/z/gkbmfslq2023/zytb.jsp}{https://gaokao.chsi.com.cn/z/gkbmfslq2023/zytb.jsp}, accessed 08/14/2024. In provinces such as Guizhou and Shaanxi, admissions to specialized institutions are via the Boston mechanism. Also, in practice, students can only report a limited length of preferences within each tier. These variations and restrictions are not modeled in this paper.} 
In Turkey, high schools are partitioned into two tiers: private schools first and public schools second. In both tiers, school priorities are determined by a weighted average of nationwide entrance exam scores and their Grade Point Average.\footnote{In Turkey, students apply to private schools through a decentralized process before applying to public schools via a centralized system. The decentralized application process for private schools is not modeled in this paper. See \cite{andersson-2024} for more details of this mechanism.} 
Similar mechanisms are also used for school admissions in the US \citep{dur-2018} and Sweden \citep{andersson-2017, andersson-2024}, for hiring public officers in Japan \citep{hatakeyama-2024}, and for appointing state school teachers in Turkey \citep{dur-2018}. These mechanisms are discussed in more detail in \Cref{sec:lit}.

A natural question is whether introducing tiers to the DA mechanism can improve the welfare of top-tier schools. This is often a goal of the TDA mechanism, especially when market designers can decide whether and how to set a tier structure.
In China, for example, authorities claim that introducing a tier structure helps ``ensure the quality of students admitted to top-tier programs'' \citep{li-2020}.\footnote{In its original context, this statement refers to the motivation of dividing academic universities into key and other universities, with ``top-tier schools'' denoting key universities rather than specialized institutions.}
In line with this, to ``alleviate talent shortages'' in certain fields, a 2020 reform added a new tier before the existing three and placed related majors in this tier.\footnote{This 2020 reform refers to the ``strengthening basic disciplines plan.'' It creates an extra tier (tier 0) for majors in basic disciplines at 36 top universities. Schools hold extra exams for these majors to determine their priorities. Students can apply to only one school in this tier. We further discuss the effect of this reform in \Cref{sec:incomplete}. For more information regarding the ``strengthening basic disciplines plan,'' see Shuo Zou, \textit{Enrollment plan targets basics}, China Daily, January 16, 2020, \href{http://english.www.gov.cn/statecouncil/ministries/202001/16/content_WS5e1fbffcc6d0891feec02516.html}{http://english.www.gov.cn/statecouncil/ministries/202001/16/content\_WS5e1fbffcc6d0891feec02516.html}, accessed 08/14/2024.}
Besides, even in markets where tiering is inevitable, such as those in Turkey, Japan, and the US, this question helps determine the order of the tiers and which schools to put in each tier.

Indeed, if students report their preferences truthfully, this question has a positive answer. This is because top-tier schools now face less competition from other schools. The problem is that, as shown in \Cref{exp:TDA}, students do not have an incentive to report truthfully under the TDA mechanism.\footnote{In \Cref{sec:st-pf}, we characterize the preference domain under which the TDA mechanism is strategy-proof.}
The necessity for strategic misreporting has been highlighted by the Chinese media, as shown in the following quote from Taizhou Daily:\footnote{Lingjia Yan, \textit{Early Tier Applications Require a Dialectical Approach}, Taizhou Daily, June 28, 2016, \href{https://www.tzc.edu.cn/info/1072/44958.htm}{https://www.tzc.edu.cn/info/1072/44958.htm}, accessed 08/14/2024.}
\begin{quote}
    ``When filling in early tier applications, students must be cautious, especially those with high scores. \ldots. Think carefully about whether you are truly keen on these top-tier schools, or if you are more inclined towards schools in the later rounds.''
\end{quote}
In contrast, the DA mechanism is strategy-proof, i.e., reporting truthfully is a dominant strategy for students.

As we will also show in \Cref{exp:TDA}, when students misreport their preferences, the TDA mechanism may cause top-tier schools to be worse off. 
That is, comparing to the dominant strategy equilibrium of the DA mechanism, lower-ranked students are admitted to top-tier schools in the equilibrium of the preference revelation game induced by the TDA mechanism. 
Although our analysis focuses mainly on complete information settings, this negative result is robust to incomplete information settings, where students are uncertain about others' preferences or schools' priorities (see \Cref{sec:incomplete} for examples).

\begin{example} \label{exp:TDA}
    Consider three students, $1, 2, 3$, and three schools, $a, b, c$, each with one seat. School $a$ is in tier 1, and schools $b$ and $c$ are in tier 2. Student preferences, $R_1, R_2, R_3$, as well as school priorities, $\succsim_a, \succsim_b, \succsim_c$, are as follows:
    \begin{multicols}{3}
    \begin{center}
    \begin{tabular}{c|c|c}
        \textbf{$R_1$} & \textbf{$R_2$} & \textbf{$R_3$} \\
        \hline
        $c$ & $b$ & $b$ \\
        $b$ & $c$ & $a$ \\
        $a$ & $a$ & $c$ \\
        $1$ & $2$ & $3$ \\
    \end{tabular}
    \end{center}
   
    \columnbreak
    \begin{center}
    \begin{tabular}{c|c|c}
        \textbf{$\succsim_a$} & \textbf{$\succsim_b$} & \textbf{$\succsim_c$} \\
        \hline
        $1$ & $1$ & $3$ \\
        $3$ & $2$ & $1$ \\
        $2$ & $3$ & $2$ \\
    \end{tabular}
    \end{center}

    \columnbreak
    \begin{center}
    \begin{tabular}{c|c|c}
        \textbf{$Q_1$} & \textbf{$Q_2$} & \textbf{$Q_3$} \\
        \hline
        $c$ & $b$ & $b$ \\
        $b$ & $c$ & $c$ \\
        $1$ & $a$ & $3$ \\
        $a$ & $2$ & $a$ \\
    \end{tabular}
    \end{center}
    \end{multicols}
    First, we show that the TDA mechanism is not strategy-proof. 
    Suppose students report truthfully, the TDA outcome is $((1, a), (2, b), (3, c))$. However, student $1$ can benefit by misreporting $Q_1$ and obtaining $((1, c), (2, b), (3, a))$ if students $2$ and $3$ still report truthfully. This shows that under the TDA mechanism, when students prefer schools in later tiers, not applying to top-tier schools can be a profitable deviation.
    
    Second, we show that the TDA mechanism can result in lower-ranked students at top-tier schools compared to the DA mechanism. 
    Let us consider the preference revelation game induced by the TDA mechanism. One Nash equilibrium of this game is $Q$, with the TDA equilibrium outcome $((1, c), (2, a), (3, b))$. Under the DA mechanism, assuming everyone reports truthfully, the DA outcome is $((1, c), (2, b), (3, a))$. Now, note that under the DA mechanism, the top-tier school $a$ gets student $3$, but in a Nash equilibrium outcome under the TDA mechanism, $a$ gets the lower-ranked student $2$.\footnote{Let us restate this example using the notation that we defined later in \Cref{sec:model}. Here, the tier structure is $t := (1, 2, 2)$. The TDA Nash equilibrium outcome under $Q$, $((1, c), (2, a), (3, b))$, is denoted by $TDA(t)(Q)$. The DA outcome (i.e., the SOSM), $((1, c), (2, b), (3, a))$, is denoted by $DA(R)$.}

    In fact, if we look at all possible equilibria under the TDA mechanism, there are only two equilibrium outcomes: $((1, c), (2, a), (3, b))$ and $((1, c), (2, b), (3, a))$. Thus, in this example, school $a$ is never strictly better off in any equilibrium under the TDA mechanism.\footnote{Even when we consider mixed Nash equilibrium outcomes, school $a$ still cannot be strictly better off. This is because listing $a$ as unacceptable is student 1's dominant strategy under the TDA mechanism.}$\hfill\square$
\end{example}

We have four main results. The first three results analyze the Nash equilibrium outcomes under the TDA mechanism, as summarised in \Cref{fig:venn}. 
These results lead to the fourth result, which formalizes the second message from \Cref{exp:TDA}: contrary to its intention, switching from the DA to the TDA mechanism may harm top-tier schools under most tier structures.
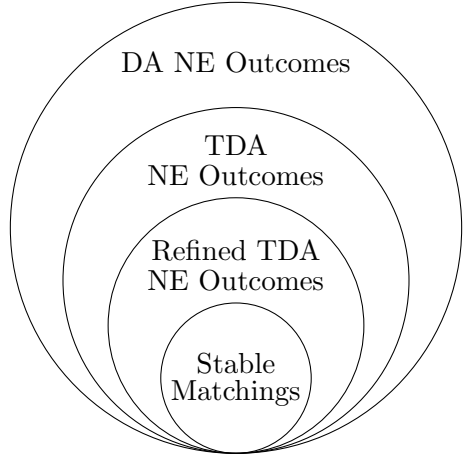
\begin{figure}
    \centering
    \begin{tikzpicture}[global scale=1]
        \draw (0,0) circle [radius=3cm]; 
        \draw (0,-0.7) circle [radius=2.3cm]; 
        \draw (0,-1.3) circle [radius=1.7cm]; 
        \draw (0,-2) circle [radius=1cm]; 
        \node[font=\small] at (0,2.2) {DA NE Outcomes};
        \node[font=\small] at (0,1.1) {TDA};
        \node[font=\small] at (0,0.7) {NE Outcomes};
        \node[font=\small] at (0,-0.3) {Refined TDA};
        \node[font=\small] at (0,-0.7) {NE Outcomes};
        \node[font=\small] at (0,-1.8) {Stable};
        \node[font=\small] at (0,-2.2) {Matchings};
    \end{tikzpicture}
    \caption{Set inclusion relation when within-tier acyclicity is not satisfied}
    \label{fig:venn}
\end{figure}

Our first result shows that all stable matchings are TDA Nash equilibrium outcomes (\Cref{thm:nested}). A matching is stable if no student prefers a different school over his current match unless that school is already filled with higher-priority students. Our finding generalizes \posscite{romeromedina-1998} result that all stable matchings are DA Nash equilibrium outcomes. However, not all equilibrium outcomes under the TDA mechanism are stable.

Our second result shows the nested structure of the set of TDA Nash equilibrium outcomes: all TDA equilibrium outcomes under a finer tier structure are equilibrium outcomes under a coarser tier structure (\Cref{thm:nested}). This result indicates that if undesirable equilibrium outcomes exist, merging tiers alone cannot resolve the issue.

Our third result identifies within-tier acyclicity as a necessary and sufficient condition for the TDA mechanism to implement stable matchings in equilibrium (\Cref{thm:implementation}). This condition weakens \posscite{ergin-2002} acyclicity definition, requiring that acyclicity only holds within each tier. Then, the finding builds on \posscite{haeringer-2009} insight that DA's Nash equilibrium outcomes are stable if and only if acyclicity is satisfied. Our result provides a further critique of the TDA mechanism. If the tier structure satisfies within-tier acyclicity—implying the priorities within each tier are sufficiently similar—then changing the order of tiers does not affect the equilibrium outcomes. This means that schools in earlier tiers do not gain any advantage from being ranked higher than other schools.

Our fourth result provides a unique solution to guarantee the TDA mechanism benefits certain schools: each tier must contain only one school (\Cref{thm:unique}). Formally, if this condition is not met, then for any school that we aim to benefit, there exists a school choice problem (a combination of quotas, priorities, and preferences) under which at least one TDA equilibrium leads to a worse outcome than under DA's dominant strategy equilibrium. Under this tier structure, all schools are weakly better off, and the ranking of tiers does not affect the equilibrium outcomes. 
However, since such a tier structure is impractical, we conclude that the TDA mechanism harms the schools that we want to benefit in general.

The rest of the paper is organized as follows. \Cref{sec:model} defines the school choice model and the DA and TDA mechanisms. \Cref{sec:nested} analyzes the Nash equilibrium outcomes under the TDA mechanism and its nested structure. \Cref{sec:implementation} studies the implementability of stable matchings. \Cref{sec:goal} analyzes the welfare of schools under the TDA mechanism and provides a unique solution for achieving TDA's goal. \Cref{sec:dis} discusses several further topics. \Cref{sec:lit} reviews related literature, and \Cref{sec:conclude} concludes. All proofs are relegated to \Cref{apx:A}. Additional examples are provided in \Cref{apx:B}.

\section{Model} \label{sec:model}
Following \cite{abdulkadiroglu-2003}, we define a \textbf{school choice problem} $E$ to be a 5-tuple $(I, S, q, \succsim, R)$ which includes:
\begin{itemize}
    \item $I$: a finite set of \textbf{students}.
    \item $S$: a finite set of \textbf{schools}.
    \item $q = (q_s)_{s \in S}$: a vector of \textbf{quotas} for each school, where $q_s \in \mathbb N$.
    \item $\succsim = (\succsim_s)_{s \in S}$: a \textbf{priority profile}, with each $\succsim_s$ being a linear order over $I$. $\succ_s$ denotes the strict relation.
    \item $R = (R_i)_{i \in I}$: a \textbf{preference profile}, with each $R_i$ being a linear order over $S \cup \{i\}$. $P_i$ denotes the strict relation.
\end{itemize}

Throughout the paper, $I$ and $S$ are fixed. Thus, unless otherwise stated, we shorthand a school choice problem $E$ as $(q, \succsim, R)$.

Let $\mathcal{R}$ denote the set of linear orders over $S \cup \{i\}$, i.e., the set of possible preferences of a student. Unless otherwise specified, when referring to a preference, $R_i$, the implicit domain is $\mathcal{R}$.

A \textbf{truncated preference} $R^{S^\prime}_i$ with respect to a subset of schools $S^\prime \subseteq S$ is a linear order over $S^\prime \cup \{i\}$ which is consistent with the original preference $R_i$, i.e., $s \mathrel{R^{S^\prime}_i} s^\prime \iff s \mathrel{R_i} s^\prime$ for any $s, s^\prime \in S^\prime \cup \{i\}$.

A set of schools $S$ is \textbf{tiered} if schools are partitioned to ${\{S_k\}}_{k \in \{1, \ldots, T\}}$, where each $S_k \subseteq S$ is a nonempty subset of schools and $T \in \mathbb N_{+}$ is the total number of partitions. We denote $S_k$ as tier $k$ and any school $s \in S_k$ as a tier $k$ school. 
A \textbf{tier structure} is defined by $t = (t_s)_{s \in S}$, where $t_s \in \{1, \ldots, T\}$ denotes the tier to which school $s$ belongs, i.e., $s \in S_{t_s}$.\footnote{Alternatively, we can define tier $k$ using the tier structure $t$ as $S_k := \{s \in S: t_s = k \}$.}

A tuple consisting of quotas and priority profile $(q, \succsim)$ is referred to as a \textbf{priority structure}. A tuple consisting of quotas, priority profile, and tier structure $(q, \succsim, t)$ is referred to as a \textbf{generalized priority structure}.

The outcomes of the school choice problem are matchings. A \textbf{matching} is a function $\mu: I \rightarrow S \cup I$ satisfying $\mu(i) \in S \cup \{i\}$ for any student $i \in I$ and $|\mu^{-1}(s)| \le q_s$ for any school $s \in S$.

A matching $\mu$ is \textbf{stable} with respect to preferences $R$ if it satisfies: (1) \textbf{individual rationality}: $\mu(i)\mathrel{R_i}i$ for any student $i \in I$; (2) \textbf{no justified envy}: there are no students $i, j \in I$ and school $s \in S$ such that $s = \mu(j)$, $s \mathrel{P_i} \mu(i)$, and $i \succ_s j$; (3) \textbf{non-wastefulness}: there are no student $i \in I$ and school $s \in S$ such that $s \mathrel{P_i} \mu(i)$ and $q_s > |\mu^{-1}(s)|$. If condition (2) or (3) is violated, the pair $(i, s)$ containing the corresponding student $i$ and school $s$ in the above definition is called a \textbf{blocking pair}. In other words,
a blocking pair is a pair of a student and a school who want to be matched with each other rather than staying with their current assignments.
The set of all stable matchings with respect to preferences $R$ in school choice problem $E=(q, \succsim, R)$ is denoted by $\mathcal{S}(E)$.

A \textbf{mechanism} $f$ is a function that assigns a matching $f(E)$ for every school choice problem $E = (q, \succsim, R)$, where $f(E)(i)$ is the school (or empty seat) assigned to student $i$. When the priority structure $(q, \succsim)$ is fixed, we denote only the preferences $R$ as the input of a mechanism.

A mechanism $f$ is \textbf{strategy-proof} if, for any school choice problem $E = (q, \succsim, R)$, $i \in I$ and $Q_i \in \mathcal{R}$, $f(q, \succsim, R)(i) \mathrel{R_i} f(q, \succsim, Q_i, R_{-i})(i)$.

\subsection{Mechanisms} \label{sec:mech}
To formally define the DA and TDA mechanisms, we first define the DA algorithm.

\subsubsection*{The Deferred Acceptance Algorithm \citep{gale-1962}}
\begin{quote}
    The algorithm takes $(I, S, q, \succsim, R)$ as the input.

    \textit{Step 1:} Each student $i \in I$ proposes to his best alternative in $S \cup \{i\}$ according to $R_i$. Each school $s$ tentatively accepts the $q_s$-highest ranked students according to $\succsim_s$, among those that have proposed to it.

    \textit{Step $k > 1$:} Each student $i \in I$ who has not been tentatively accepted in Step $k-1$ proposes to his best alternative in $S \cup \{i\}$ according to $R_i$, among those to which he has not previously proposed. Each school $s$ tentatively accepts the $q_s$-highest ranked students according to $\succsim_s$, among those that have proposed to it in this step or were tentatively accepted by $s$ in Step $k-1$.

    The algorithm terminates when no new proposals are made, and outputs the final matching.
\end{quote}

The \textbf{deferred acceptance mechanism} takes a school choice problem $(q, \succsim, R)$ as the input, and outputs the outcome of the DA algorithm under $(I, S, q, \succsim, R)$.

We use $DA$ to denote the DA mechanism. When there is no ambiguity in the priority structure $(q, \succsim)$, we shorthand the outcome under a preference profile $R$ as $DA(R)$.

The \textbf{student-optimal stable matching} (SOSM) is a stable matching where all students weakly prefer their outcomes over any other stable matchings. However, all schools weakly prefer other stable matchings to the SOSM. The SOSM can be obtained by the DA mechanism when everyone is reporting truthfully \citep{gale-1962}.

\subsubsection*{The Tiered Deferred Acceptance Mechanism under the Tier Structure $t$}

\begin{quote}
The mechanism takes a school choice problem $(q, \succsim, R)$ as the input.

The outcome $\mu$ is determined via the following algorithm:

\textit{Round 1:} Run the DA algorithm under $(I, S_1, (q_s)_{s \in S_1}, (\succsim_s)_{s \in S_1}, (R^{S_1}_i)_{i \in I}$, and denote the output by $\mu_1$. That is, each student $i \in I$ proposes to his acceptable schools in tier 1 according to his truncated preference for them. For any student $i \in I$, if $\mu_1(i) \neq i$, then $\mu(i) = \mu_1(i)$.

\textit{Round $k > 1$:} Denote the set of students who are unmatched after Round $k-1$ as $I_k := \{i \in I: \mu_{k-1}(i) = i\}$. Run the DA algorithm under $(I_k, S_k, (q_s)_{s \in S_k}, (\succsim_s)_{s \in S_k}, (R^{S_k}_i)_{i \in I_k}$, and denote the output by $\mu_k$. That is, each currently unmatched student $i \in I_k$ proposes to his acceptable schools in tier $k$ according to his truncated preference for them. For any student $i \in I_k$, if $\mu_k(i) \neq i$, then $\mu(i) = \mu_k(i)$.

The algorithm terminates after round $T$. Then, for any $i \in I_T$, $\mu(i) = \mu_T(i)$.

The mechanism outputs the outcome $\mu$ of the algorithm.
\end{quote}

Each tier structure $t$ induces a new TDA mechanism.
We use $TDA(t)$ to denote the TDA mechanism under the tier structure $t$. Similarly, when there is no ambiguity in the priority structure $(q, \succsim)$, we shorthand the outcome under a preference profile $R$ as $TDA(t)(R)$.

We distinguish the TDA mechanism from two similar mechanisms, the ``parallel'' mechanism \citep{chen-2017} and the iterative DA mechanism \citep{bo-2022}, in \Cref{sec:similar}.

\section{Nash Equilibria under the TDA Mechanism} \label{sec:nested}
As shown in \Cref{exp:TDA}, the TDA mechanism is not strategy-proof. To predict the equilibrium outcome, we study the induced preference revelation game. In this section, we show the equilibrium outcomes are nested with respect to the tier structure.

In this game, school priorities and tiers are fixed, while students can report their preferences strategically. Given any school choice problem $E = (q, \succsim, R)$, a preference profile $Q \in \mathcal{R}^{|I|}$ is a (pure) \textbf{Nash equilibrium of the preference revelation game induced by the mechanism $f$ at $E$} if, for any student $i \in I$ and preference $Q^\prime_i \in \mathcal{R}$, $f(Q_i, Q_{-i})(i) \mathrel{R_i} f(Q^\prime_i, Q_{-i})(i)$.

To avoid ambiguity, in the discussion of Nash equilibrium, we use $R$ to denote the true preference profile defined in the school choice problem $E = (q, \succsim, R)$, and $Q$ to denote an arbitrary preference profile. The same notation applies to the individual preferences $R_i$ and $Q_i$. The implicit domain for $Q_i$ is $\mathcal{R}$. 

Let us denote $\mathcal E^{f}(E)$ as the set of all Nash equilibria in the preference revelation game induced by mechanism $f$ at school choice problem $E = (q, \succsim, R)$, and $\mathcal O^{f}(E)$ as the set of Nash equilibrium outcomes under $f$ at $E$, i.e., $\mathcal O^{f}(E) := \{f(q, \succsim, Q): Q \in \mathcal E^{f}(E) \}$.

We show two results in \Cref{thm:nested}. The first one is that stable matchings are TDA Nash equilibrium outcomes. This result confirms the existence of Nash equilibria under the TDA mechanism with any tier structure.

Before presenting our second result, we define a partial order over tier structures. A tier structure is a refinement of another, if the original partition of schools is refined, and the original order of tiers is maintained.

\begin{definition}
    A tier structure $t$ is a \textbf{refinement} of another tier structure $t^\prime$ if, for any two schools $a, b \in S$, $t^\prime_a > t^\prime_{b}$ implies $t_a > t_{b}$.
\end{definition}

We are now ready to present the second result. If a matching is a TDA Nash equilibrium outcome under the refined tier structure, then it is also a Nash equilibrium outcome under the coarser tier structure, i.e., TDA Nash equilibrium outcomes are nested with respect to the partial order defined over the tier structure.

\begin{theorem} \label{thm:nested}
    For any school choice problem $E$ and tier structure $t$, the following holds.
    \begin{enumerate}
        \item The set of stable matchings under the preferences in $E$ is a subset of the set of Nash equilibrium outcomes under the TDA mechanism with $t$.
        \item For any tier structure $t^\prime$ such that $t$ is a refinement of $t^\prime$, the set of Nash equilibrium outcomes under the TDA mechanism with $t$ is a subset of the set of Nash equilibrium outcomes under the TDA mechanism with $t^\prime$, i.e., 
    \end{enumerate}
    $$\mathcal S(E) \subseteq \mathcal O^{TDA(t)}(E) \subseteq \mathcal O^{TDA(t^\prime)}(E).$$
\end{theorem}

We now provide a proof sketch for \Cref{thm:nested}. An intuition is provided in \Cref{rmk:intuition}. 
First, we sketch the proof of part 1 of \Cref{thm:nested}. For stable matching $\mu$, let each student $i$ report only his assignment $\mu(i)$ as acceptable. It can be checked that this constructed preference profile is a Nash equilibrium under the TDA mechanism. The reasoning is further explained in \Cref{rmk:stable}. 

Second, we sketch the proof of part 2 of \Cref{thm:nested} in three steps.\footnote{The three steps in this proof sketch are not the three steps in the proof in \Cref{apx:A}. Specifically, steps 1 and 2 here correspond to step 1 in \Cref{apx:A}, while step 3 here corresponds to steps 2 and 3 there.} 
In the first step, we provide a condition for the TDA mechanism to be outcome-equivalent to the DA mechanism. This condition requires the reported preferences to be aligned with the tier structure, meaning that students always prefer schools in earlier tiers over those in later tiers. The result is true because under an aligned preference profile, in each step the two mechanisms yield the same outcome.\footnote{We further discuss the alignment between preferences and tier structure in \Cref{sec:st-pf}.}

In the second step, we find a further connection between the outcomes of the two mechanisms. We construct a mapping $(Q, t) \mapsto \widetilde Q$, which reshuffles preferences $Q$ into $\widetilde Q$ according to the tier structure, i.e., making each reshuffled preference $\widetilde Q_i$ aligned with $t$ while maintaining the relative ranking of schools inside each tier. Two implications follow. First, as the TDA mechanism only uses preference inside each tier, we have $TDA(t)(Q) = TDA(t)(\widetilde Q)$. Second, since each $\widetilde Q_i$ is aligned with $t$, we have $TDA(t)(\widetilde Q) = DA(\widetilde Q)$. Thus, we obtain the desired connection: $TDA(t)(Q) = DA(\widetilde Q)$. 

In the third step, we show that any reshuffled equilibrium $\widetilde Q$ under a finer tier structure $t$ is an equilibrium under a coarser tier structure $t^\prime$, i.e., $Q \in \mathcal{E}^{TDA(t)}(E) \Rightarrow \widetilde Q \in \mathcal{E}^{TDA(t^\prime)}(E)$. This result directly implies part 2 of \Cref{thm:nested}. The proof is by showing $\widetilde Q$ is, in fact, a Nash equilibrium under the DA mechanism. Then, since $\widetilde Q$ is aligned with both $t$ and $t^\prime$, we show it is also a TDA equilibrium under the coarser tier structure $t^\prime$.

\begin{remark}\label{rmk:aligned}
    When students' true preferences align with the tier structure, the TDA mechanism is strategy-proof. This is because the TDA mechanism is outcome-equivalent to the DA mechanism under this condition.
    In fact, as we will show in \Cref{sec:st-pf}, this condition is also necessary for the TDA mechanism to be strategy-proof (\Cref{prop:st-pf}).$\hfill\square$
\end{remark}

\begin{remark}\label{rmk:stable}
    The TDA mechanism can be viewed as a special version of the DA mechanism, where students are limited to reporting preferences that align with the tier structure.
    This observation explains why all stable matchings are TDA equilibrium outcomes. \cite{haeringer-2009} show that for any stable matching, each student reporting only his assignment as acceptable consists of a Nash equilibrium under the DA mechanism. Then, since such an equilibrium meets the preference restrictions of the TDA mechanism, it is also a Nash equilibrium under the TDA mechanism.\footnote{Note that this observation does not imply all reshuffled TDA Nash equilibria are DA Nash equilibria. This is because the best response in a restricted preference domain may not be the best response in an unrestricted domain.}$\hfill\square$
\end{remark}

One implication for \Cref{thm:nested} is that there is a limitation in selecting Nash equilibrium outcomes. \Cref{thm:nested} suggests that some equilibrium outcomes can be eliminated by refining the tier structure. However, this method only applies to unstable outcomes. Thus, if a market designer wants to get rid of a stable matching in equilibrium, adjusting the tier structure cannot achieve this goal.

\section{Implementation of Stable Matchings} \label{sec:implementation}
As shown in \Cref{exp:TDA}, not all TDA Nash equilibrium outcomes are stable.\footnote{For example, in the Nash equilibrium outcome $TDA(t)(Q)$, $(2, b)$ forms a blocking pair. This observation aligns with \posscite{sotomayor-2008} example that not all equilibrium outcomes under the DA mechanism are stable.} In this section, we introduce the within-tier acyclicity condition, which ensures the stability of all TDA Nash equilibrium outcomes. 

This condition is derived from the acyclicity condition first introduced by \cite{ergin-2002}. Simply put, if a priority structure is acyclic, then no student can block others' assignments without affecting his assignment.

\begin{definition}\label{def:acyclicity}
    Given priority structure $(q, \succsim)$, a \textbf{cycle} consists of distinct schools $a, b \in S$ and students $i, j, k \in I$ such that:
    \begin{enumerate}
        \item $i \succ_a j \succ_a k$ and $k \succ_{b} i$, and
        \item there exist (possibly empty) disjoint sets of students $I_a, I_b \subset I \setminus \{i, j, k\}$, such that $I_a \subset U_a(j)$, $I_b \subset U_b(i)$, $|I_a| = q_a - 1$, and $|I_b| = q_b - 1$, where $U_s(i) := \{j \in I: j \succ_s i \}$.
    \end{enumerate}
    A priority structure $(q, \succsim)$ is \textbf{acyclic} if it has no cycles. 
\end{definition}

We further weaken the acyclicity condition to within-tier acyclicity, which only requires there are no cycles among schools inside each tier.  

For any tier $k \in \{1, \ldots, T\}$, we denote the quotas and priorities of tier $k$ schools by $q_{S_k} := (q_s)_{s \in S_k}$ and $\succsim_{S_k} := (\succsim_s)_{s \in S_k}$, respectively.

\begin{definition}\label{def:withinacyclicity}
    A generalized priority structure $(q, \succsim, t)$ is \textbf{within-tier acyclic} if, for any tier $k \in \{1, \ldots, T\}$, the priority structure $(q_{S_k}, \succsim_{S_k})$ is acyclic.
\end{definition}

Our third result shows that within-tier acyclicity is a necessary and sufficient condition for the TDA mechanism to implement the stable correspondence in Nash equilibria. This result generalizes \posscite{haeringer-2009} finding that acyclicity is the necessary and sufficient condition for the DA mechanism to implement stable matchings in Nash equilibria. 

\begin{theorem}\label{thm:implementation}
    A generalized priority structure $(q, \succsim, t)$ is within-tier acyclic if and only if, for any school choice problem $E$ with $(q, \succsim)$, the set of stable matchings under the preferences in $E$ is equal to the set of Nash equilibrium outcomes under the TDA mechanism with $t$, i.e., $$\mathcal{S}(E) = \mathcal{O}^{TDA(t)}(E).$$
\end{theorem}

We provide an intuition for the ``only if'' direction in three steps. That is, the existence of an unstable TDA Nash equilibrium outcome implies cycles in some tiers.

In the first step, we show that the instability of Nash equilibrium outcomes can only come from justified envy. This statement follows from \posscite{haeringer-2009} result that all DA Nash equilibrium outcomes are individually rational and non-wasteful. Combining it with \Cref{thm:nested}, we know for any unstable TDA Nash equilibrium outcome $\mu$, there exists a blocking pair $(i,s)$. 

In the second step, we show the blocking pair comes from a (potential) rejection chain. Let us consider what happens if student $i$ tries to ``realize'' this blocking pair. That is, $i$ reports another preference where he can propose to $s$ before being finally matched. Since $\mu$ is a Nash equilibrium outcome, $i$ must be rejected by $s$ in the final matching. Thus, when $i$ proposes to $s$, he is either rejected immediately or tentatively accepted, with the latter case triggering a rejection chain that eventually leads to $i$ being rejected by $s$.\footnote{Let us use \Cref{exp:TDA} to see what is a rejection chain. Under $TDA(t)(Q)$, $(2,b)$ forms a blocking pair, but when student $2$ only reports school $b$ as acceptable, $b$ rejects student $3$, who then applies to $c$, causing $c$ to reject student 1. Subsequently, student 1 applies to $b$, leading $b$ to reject student 2. See \cite{kesten-2010} for more examples and applications of rejection chains.}

In the third step, we note two points regarding this rejection chain. First, the rejection chain can only involve schools within $s$'s tier. Suppose otherwise the chain involves a school $\hat s$ in a later tier; then once it reaches $\hat s$, the matching for $s$ has been finalized, and $i$ cannot be rejected again. 
Second, the rejection chain comes from a cyclic priority structure. This is because acyclicity implies non-bossiness \citep{ergin-2002}, and the existence of a rejection chain violates the latter.

\begin{remark} \label{rmk:intuition}
    Let us provide an intuition for \Cref{thm:nested} using the argument above. We answer two questions here: (1) why introducing the tier structure eliminates some equilibrium outcomes; (2) why introducing the tier structure does not create more equilibrium outcomes.

    For the first question, the above argument applies. We know that unstable matchings come from within-tier cycles, and refining the tier structure eliminates some of these cycles. Therefore, the unstable matchings that come from these cycles are eliminated as well.

    For the second question, however, this argument does not directly apply. This is because, under a given school choice problem, not all cycles can form blocking pairs in equilibrium. Thus, we turn to the proof of \Cref{thm:nested}, where we show that any reshuffled equilibrium under the finer tier structure is also an equilibrium under the coarser tier structure. Combining the fact that this reshuffled equilibrium yields the same outcome under both tier structures, we know if a cycle forms a blocking pair through an equilibrium strategy profile under the finer tier structure, it can form the same blocking pair by the reshuffled equilibrium strategy profile (which is an equilibrium here) under the coarser tier structure. Therefore, if an additional unstable TDA equilibrium outcome exists under the finer tier structure, we can first identify the cycles that cause the instability, then use these cycles to replicate the same equilibrium outcome under the coarser tier structure - a contradiction.$\hfill\square$
\end{remark}

We show a corollary of \Cref{thm:implementation}. When within-tier acyclicity is satisfied, the choice of tier structure is irrelevant to the set of TDA Nash equilibrium outcomes, because the latter always equals the set of stable matchings.

\begin{corollary}
    \label{coro:relabel}
    For any school choice problem $E$ and tier structures $t, t^\prime$, if $(q, \succsim, t)$ and $(q, \succsim, t^\prime)$ satisfy within-tier acyclicity, then the 
    sets of Nash equilibrium outcomes under both the TDA mechanism with $t$ and the TDA mechanism with $t^\prime$ are equal to the set of stable matchings under the preference in $E$, i.e., 
    $$\mathcal{O}^{TDA(t)}(E) = \mathcal{O}^{TDA(t^\prime)}(E) = \mathcal{S}(E).$$
\end{corollary}

A special case of \Cref{coro:relabel} is that if within-tier acyclicity is satisfied, then as long as the partition of schools remains unchanged, re-ranking the tiers does not affect the TDA equilibrium outcomes.\footnote{When within-tier acyclicity is not satisfied, we conjecture that re-ranking the tiers does not affect the TDA  equilibrium outcomes.} This implies that placing some schools in the top tier provides no advantage over placing them in a lower tier.\footnote{In the Turkish high school admission system, all schools within one tier rank students using the same exam score. Therefore, all TDA equilibrium outcomes are stable, and re-ranking the tiers does not affect the outcome. This result is consistent with \posscite{andersson-2024} Proposition 5. The same result holds for the Turkish public school teacher assignment system; see \cite{dur-2018} Corollary 4 for details.}

\begin{remark}\label{rmk:tiered other}
    One may wonder what the Nash equilibrium outcomes are under other sequential mechanisms that use the Boston or serial dictatorship (SD) mechanism within each tier. \cite{dur-2018} show that the set of Nash equilibrium outcomes for the tiered Boston or tiered SD mechanism equals the set of stable matchings.\footnote{Although \cite{dur-2018} assume only two tiers, their proof easily generalizes to multiple tiers.}$\hfill\square$
\end{remark}

\section{Goal Setting: Welfare Analysis of Schools} \label{sec:goal}
As shown in \Cref{exp:TDA}, the TDA mechanism does not always benefit top-tier schools in equilibrium. This section aims to address this problem. We will define the goal of the TDA mechanism as to guarantee certain schools to be weakly better off, and provide the unique class of tier structures that can achieve this goal.

To formally define the welfare of schools, we relate the school choice problem to the well-known two-sided matching markets \citep{gale-1962}. One key difference between the two is that in the latter, both students and schools are treated as agents with preferences, whereas, in the former, only students are treated as such. By treating school priorities as their preferences, one can find counterparts in the school choice problem for concepts and findings in two-sided matching.

Following this idea, we define schools' preferences over sets of students, while ensuring consistency with the priorities over individual students. Let $\underbar{\text{$\triangleright$}}$ denote the preference profile over subsets of students and $\triangleright$ denotes the strict relation. 
A preference $\underbar{\text{$\triangleright$}}_s$ is \textbf{responsive} if, for any subset of students $J$ and any two students $i, j \not \in J \subseteq I$, $i \succ_s j$ implies $(J \cup \{i\}) \triangleright_s (J \cup \{j\})$ \citep{roth-1985}.
In this section, we assume schools have responsive preferences.\footnote{We do not require schools' preference profile $\underbar{\text{$\triangleright$}}$ to be responsive for other results in this paper. Although we use responsiveness in the proofs of \Cref{thm:nested,thm:implementation}, it is for simplifying the notation. To be more specific, when we write $\mu^{-1}(s) \mathrel{\underbar{\text{$\triangleright$}}_s} \mu^{\prime -1}(s)$, we mean that school $s$ is getting a set of worse-off students under $\mu^{\prime}$ than under $\mu$ if $s$ has responsive preference. Here, we do not require $s$ to have a responsive preference for real; we only use this concept as a concise way to describe the relation between matchings $\mu$ and $\mu^\prime$.}

As stated in the Introduction, one goal of introducing the tier structure is to guarantee that certain predetermined schools (e.g., top-tier schools) are weakly better off than under the SOSM.\footnote{While we focus on the welfare effects of introducing a tier structure in this paper, we note that moving a school from a later tier to an earlier tier can also make the school worse off in equilibrium (see \Cref{exp:move} in \Cref{apx:B}).} Below, we formalize this ``guarantee'' property.

\begin{definition}\label{def:guarantee}
    A tier structure $t$ \textbf{guarantees school $s$ to be weakly better off} if, for any school choice problem $E$, $s$ weakly prefers any Nash equilibrium outcome under the TDA mechanism with $t$ to the SOSM.
\end{definition}

Let us justify this formalization. We want to model a market designer who needs to decide on the tier structure before knowing school quotas, priorities, and student preferences. And he considers the worst-case scenario when making this decision, caring about whether there exists an equilibrium outcome under which schools in $S^*$ are made worse off. This scenario aligns with the Chinese admission system. Given that the system faces millions of new students each year, most top-tier schools' priorities are determined by school-specific exam scores, and schools have the autonomy to adjust their quotas, it is reasonable to assume $(q, \succsim, R)$ changes every year.\footnote{For more information about quota adjustment, see Jiang Zhu et al., \textit{Ministry of education releases 24 new undergraduate majors, pilot enrollment to start with this year's college entrance exam}, Chinese Central Television News, March 30, 2024, \href{https://news.cnr.cn/native/gd/20240330/t20240330_526645484.shtml}{https://news.cnr.cn/native/gd/20240330/t20240330\_526645484.shtml}, accessed 08/14/2024.}

\begin{remark}
    In \Cref{def:guarantee}, we compare the Nash equilibria under the TDA mechanism with the undominated equilibria under the DA mechanism. For a fairer comparison, we will compare the undominated equilibria under both mechanisms in \Cref{sec:undominated}.$\hfill\square$
\end{remark}

In most cases, TDA's goal cannot be achieved. In fact, as we will show in \Cref{thm:unique}, only a unique partition of schools enables the TDA mechanism to achieve it.

\begin{theorem}
    \label{thm:unique}
    If $|I| \ge 3$, for any nonempty $S^* \subseteq S$, $t$ guarantees any school $s \in S^*$ to be weakly better off if and only if, $t$ has only one school per tier.
\end{theorem}

We denote the class of TDA mechanisms with the tier structure that has one school per tier as the finest TDA mechanism.\footnote{The finest TDA mechanism is equivalent to the student-proposing SD mechanism. In this mechanism, schools are ordered according to the tier structure, and each school, in turn, selects its most preferred students from those who have reported it as acceptable and remain unmatched.}

We provide a proof sketch of \Cref{thm:unique}. The ``if'' direction follows from \Cref{thm:implementation}.
The ``only if'' direction is proven by construction. Suppose there exists a tier that contains more than one school, then it could either contain a school in $S^*$ or not. In the former case, we can easily construct a cycle inside this tier to make this school worse off. For the latter case where all schools in $S^*$ are not involved in a within-tier cycle, we show that cycles in other tiers could still make some schools in $S^*$ worse off - as illustrated in \Cref{exp:TDA}, where tier 1 school $a$ is made worse off by the cycle between tier 2 schools $b$ and $c$. Finally, since a cycle needs at least three students to form, we impose the constraint of $|I| \ge 3$.

\Cref{thm:unique} has two implications. First, the provided solution for achieving TDA's goal - having only one school per tier - is impractical. 
Increasing the number of tiers leads to higher administrative costs and greater communication burdens for both schools and students. Additionally, this solution might even be infeasible, as market designers may not have full control of partitioning or integrating tiers. For example, in New York City, the separation of exam and regular schools \citep{abdulkadiroglu-2005}, and in Turkey, the separation between private and public schools \citep{andersson-2024}, demonstrate how external factors limit such decisions.

Second, the choice of $S^*$ does not affect the final mechanism. Even if we allow some schools to be worse off, we are still constrained to using the finest TDA mechanism. As shown in \Cref{thm:implementation}, the set of equilibrium outcomes under the finest TDA mechanism equals to the set of stable matchings. As a result, the finest TDA mechanism benefits ``some'' schools by essentially benefiting ``all'' schools.

\begin{remark}\label{rmk:finest}
    We show the relationship between the finest TDA mechanism and the Boston mechanism. 
    The similarity is that both mechanisms implement stable matchings in Nash equilibria. This is because when each tier contains only one school, within-tier acyclicity is satisfied by definition, making all TDA Nash equilibrium outcomes stable. \cite{ergin-2006} show that the same result holds for the Boston mechanism.

    However, the two mechanisms differ in terms of welfare properties. When students report truthfully, the Boston mechanism is Pareto efficient for students, while the finest TDA mechanism is not.\footnote{The mechanism is not necessarily Pareto efficient for schools. This is because each school can only choose from students who report it as acceptable. For example, suppose there is only one school with a quota of 1, and all students report this school as unacceptable. In this case, the matching under the finest TDA mechanism is that everyone is unmatched. But the matching where this school is matched with its favorite student would be the Pareto efficient outcome for schools. However, the outcome of the finest TDA mechanism indeed lies on the Pareto frontier (for schools) of matchings that are individually rational with respect to students' reported preferences.} The key difference is that the Boston mechanism allows students to prioritize schools freely, while the finest TDA mechanism determines the order of schools exogenously.$\hfill\square$
\end{remark}

In \Cref{def:guarantee}, we impose two ``for any'' requirements: for any school choice problem and for any Nash equilibrium outcome. Next, we consider dropping one of the two requirements. 

First, we consider weakening the requirement for Nash equilibrium outcomes to ``there exists a Nash equilibrium outcome such that $s$ weakly prefers it to the SOSM.'' Note that this modified goal holds trivially under any tier structure because the SOSM is a Nash equilibrium outcome by \Cref{thm:nested}. However, if we strengthen ``weakly prefer'' to ``strictly prefer,'' then no tier structure can satisfy this goal. This is because, when all schools have the same priorities, the set of Nash equilibrium outcomes is equal to the stable set, which contains only the SOSM \citep{romero-2013, akahoshi-2014}.\footnote{In fact, \cite{akahoshi-2014} shows that in many-to-one problems, the stable set is a singleton if and only if an acyclicity condition, which is stronger than \Cref{def:acyclicity}, is satisfied.}

Second, we consider weakening the requirement for the school choice problem $(q, \succsim, R)$ to ``for any $R$.'' That is, we allow the market designer to know the priority structure $(q, \succsim)$ before deciding on the tier structure. Under this weakened goal, one obvious solution is to set a within-tier acyclic tier structure. But there can be more solutions. In fact, depending on the set of schools $S^*$ that the market designer wants to prioritize, cycles within the tiers that do not contain any school in $S^*$ might be allowed (see \Cref{exp:S^*} in \Cref{apx:B}).

\begin{remark}
    Since we have argued that the only solution for the TDA mechanism to achieve its goal is impractical, a natural question arises: Are there any practical mechanisms that can achieve this goal? The answer is yes. In fact, one such solution is the Boston mechanism. This is because the set of Nash equilibrium outcomes of the Boston mechanism equals the set of stable matchings \citep{ergin-2006}. For the same reason, replacing the DA mechanism within each tier with the Boston or SD mechanism also achieves the goal, as discussed in \Cref{rmk:tiered other}. 
    Another type of solution is discussed in \Cref{rmk:FDA}.$\hfill\square$
\end{remark}

\section{Discussion} \label{sec:dis}
This section provides several discussions. 
\Cref{sec:similar} distinguishes the TDA mechanism from two similar mechanisms.
\Cref{sec:st-pf} characterizes the preference domain under which the TDA mechanism is strategy-proof. 
\Cref{sec:property} considers whether, when the tier structure is refined, we can compare the corresponding TDA mechanisms in terms of manipulability or stability.
\Cref{sec:undominated} studies undominated Nash equilibria under the TDA mechanism.
\Cref{sec:incomplete} examines incomplete information settings where students are uncertain about others' preferences or schools' priorities.

\subsection{Relevant Mechanisms} \label{sec:similar}
In this section, we distinguish two similar mechanisms from the TDA mechanism. Both mechanisms relate to the DA mechanism and feature a sequential process.

The first mechanism is the ``parallel'' mechanism studied by \cite{chen-2017}. Under this mechanism, students are given choice-bands in which they can list several ``parallel” schools. The first round of admission applies the DA mechanism considering only the first choice-band of students, and so on. If a student is assigned to a school at the end of one round, then his assignment is finalized. 

The ``parallel'' mechanism is equivalent to the DA mechanism when the choice-band allows students to apply to an unlimited number of schools. Before 2009, about 10 provinces in China, including Beijing, Sichuan, and Gansu, used the ``parallel'' mechanism rather than the DA mechanism for the second tier of college admissions.\footnote{Data source: \href{https://gaokao.chsi.com.cn/gkxx/zjsd/201003/20100311/65911460.html}{https://gaokao.chsi.com.cn/gkxx/zjsd/201003/20100311/65911460.html}, accessed 08/29/2024.} However, as of 2023, all provinces except Inner Mongolia have adopted the DA mechanism within tiers.\footnote{As noted in \cref{ft:China}, students' reported preference length is limited. However, since most provinces allow students to list up to 40 schools and 96 majors, we model the Chinese admission mechanism within each tier as the DA mechanism. For more information on the limitation of preference reporting in China, see Peng Chen, \textit{The New ``Gaokao'' System Implemented in 8 Provinces: How to Fill Out Your College Preferences}, Guangming Daily, June 11, 2021, \href{http://www.moe.gov.cn/jyb_xwfb/xw_zt/moe_357/2021/2021_zt12/meiti/202106/t20210611_537429.html}{http://www.moe.gov.cn/jyb\_xwfb/xw\_zt/moe\_357/2021/2021\_zt12/meiti/202106/t20210611\_537429.html}, accessed 08/29/2024.}

The difference lies in the setup of choice-bands and tiers: when each choice-band allows for only one school, the ``parallel'' mechanism is equivalent to the Boston mechanism. However, if each tier contains only one school, the TDA mechanism does not match the Boston mechanism, as explained further in \Cref{rmk:finest}. This distinction arises because the ``parallel'' mechanism allows students to choose which schools to include in their choice-band, whereas the TDA mechanism assigns schools to tiers exogenously. In fact, none of the property-specific comparisons in \cite{chen-2017} extends to the TDA mechanism, as we will observe in \Cref{sec:property}.

The second mechanism is the iterative DA (IDA) mechanism, first introduced by \cite{bo-2022}. In this system, students iteratively submit a limited length of preferences over lists of schools. 
When there is only one iteration and students can submit their complete preference lists, the IDA mechanism is equivalent to the DA mechanism. 

We show two differences here. First, under the IDA mechanism, lists are determined endogenously based on student preferences, rather than exogenously. Second, the IDA mechanism involves a dynamic adjustment process: after each iteration, students resubmit their preferences, and their assignments may change. In contrast, the TDA mechanism is static in the sense that students report their preferences once, and their assignments are finalized after each tier unless they remain unmatched.

\subsection{Maximal Domain for Strategy-Proofness} \label{sec:st-pf}
In this section, we characterize the maximal preference domain for the TDA mechanism to be strategy-proof. Then, we use this result to show why the intuition that introducing a tier structure benefits top-tier schools fails.

First, we define the alignment between a preference and a tier structure. A preference aligns with a tier structure if all schools in earlier tiers are preferred to those in later tiers.\footnote{Another way of stating this condition is that the preference profile is tiered by the tier structure (see \cref{ft:tiered-pref} for the definition of a tiered preference profile). A tiered preference domain has been studied in the two-sided matching literature and has been identified as a solution to finding a stable and strategy-proof mechanism \citep{kesten-2010, akahoshi-2014b, kesten-2019, hatakeyama-2024}.}

\begin{definition}
    A preference $R_i$ is \textbf{aligned} with tier structure $t$ if for any schools $s, s^\prime \in S$, $s \mathrel{R_i} s^\prime \Rightarrow t_s \le t_{s^\prime}$.
\end{definition}

Fixing a tier structure $t$, let $\mathcal{R}^t := \{R_i \in \mathcal{R}: R_i \text{ is aligned with } t\}$ be the preference domain under which preferences are aligned with $t$.

\begin{proposition}
    \label{prop:st-pf}
    For any tier structure $t$, the TDA mechanism with $t$ is strategy-proof if and only if for any school choice problem $E = (q, \succsim, R)$, each preference $R_i$ is aligned with $t$, i.e., $R_i \in \mathcal{R}^t$ for any $i \in I$.
\end{proposition}

We sketch the proof of \Cref{prop:st-pf}. 
The ``if'' direction directly follows from \Cref{rmk:aligned}. For the ``only if'' direction, consider the following intuition: Suppose a student prefers a later tier school to an earlier tier school and has very high priority at the later tier school. In this case, the student may want to misreport the earlier tier school as unacceptable to ensure admission to his preferred school.

Let us consider an implication of \Cref{prop:st-pf}.
One intuition for placing some schools in the top tier is that it reduces competition, potentially benefiting those schools. To formalize this idea, suppose students always report truthfully, then the top-tier schools will be weakly better off.\footnote{See \Cref{lem:lessstudent} in \Cref{apx:A} for the formal proof. However, the welfare analysis for the second or third-tier schools is more ambiguous - they may not be weakly better than under the SOSM. Additionally, it is not true that if the third-tier schools are no worse off, so will the second-tier schools. Consider an example with three schools, $a, b, c$, and three students, $1, 2, 3$. Schools have the same priority $1 \succ 2 \succ 3$ and the same quota of $1$. Students have the same preference $b \mathrel{P_i} a \mathrel{P_i} c \mathrel{P_i} i$. The SOSM is $((1, b), (2, a), (3, c))$. But if the tier structure is $(1, 2, 3)$ and all students report truthfully, the TDA outcome is $((1, a), (2, b), (3, c))$. In this example, the second-tier school $b$ is strictly worse off, while the third-tier school $c$ obtains the same outcome.}
Since we assume students are fully sophisticated, one reasonable case in which they always report truthfully is when the TDA mechanism is strategy-proof. Then, as shown in \Cref{prop:st-pf}, we can only choose from tier structures that align with all possible true preferences of students. However, even if such a restrictive condition is met, the TDA mechanism will be outcome-equivalent to the DA mechanism. In other words, introducing the tier structure does not bring any new advantages to the mechanism.
Therefore, \Cref{prop:st-pf} suggests that the intuition for introducing a tier structure breaks down when considering the strategizing behavior of students. This necessitates conducting the Nash equilibrium analysis as discussed in \Cref{sec:nested,sec:implementation}.

\subsection{Property-Specific Comparison} \label{sec:property}
In this section, we study properties of the TDA mechanism, such as manipulability and stability. One natural question is: when the tier structure is refined, do monotonic relationships exist between the corresponding TDA mechanisms in terms of these properties? The answer is no.

Following the framework of \cite{pathak-2013} and \cite{chen-2017}, we define the following properties for comparison across mechanisms.

A mechanism $f$ is \textbf{manipulable} at a school choice problem $E$ if, there exists a student $i \in I$ such that there exists preference $Q_i \in \mathcal{R}$ with $f(Q_i, R_{-i})(i) \mathrel{R_i} f(R)(i)$, i.e., $R$ is not a Nash equilibrium. A mechanism $f$ is \textbf{more manipulable} than mechanism $f^\prime$ if (1) at any problem $f^\prime$ is manipulable, then $f$ is also manipulable; and (2) the converse is not always true, i.e., there is at least one problem at which $f$ is manipulable but $f^\prime$ is not.

A mechanism $f$ is \textbf{stable} at a school choice problem $E$ if $f(E)$ is stable. A mechanism $f$ is \textbf{more stable} than mechanism $f^\prime$ if (1) at any problem $f^\prime$ is stable, $f$ is also stable; and (2) the converse is not always true, i.e., there is at least one problem at which $f$ is stable but $f^\prime$ is not. 

Given that the DA mechanism is both strategy-proof and stable, we know that the TDA mechanism under a coarser tier structure is not more manipulable nor less stable than under a finer tier structure. In \Cref{exp:mani,exp:stable}, we show that the converse side is not true either. That is, the TDA mechanism under a finer tier structure is not necessarily more manipulable nor less stable than under a coarser tier structure.

\begin{example}[The TDA mechanism is not more manipulable as the tier structure becomes finer] \label{exp:mani}
    Consider three students, $1, 2, 3$, and three schools, $a, b, c$, each with one seat. Student preferences $R$ as well as school priorities $\succsim$ are as follows:
    \begin{multicols}{2}
    \begin{center}
    \begin{tabular}{c|c|c}
        \textbf{$R_1$} & \textbf{$R_2$} & \textbf{$R_3$} \\
        \hline
        $c$ & $b$ & $b$ \\
        $b$ & $c$ & $a$ \\
        $a$ & $a$ & $c$ \\
        $1$ & $2$ & $3$ \\
    \end{tabular}
    \end{center}
   
    \columnbreak
    \begin{center}
    \begin{tabular}{c|c|c}
        \textbf{$\succsim_a$} & \textbf{$\succsim_b$} & \textbf{$\succsim_c$} \\
        \hline
        $3$ & $2$ & $1$ \\
        $2$ & $1$ & $2$ \\
        $1$ & $3$ & $3$ \\
    \end{tabular}
    \end{center}
    \end{multicols}

    Under the tier structure $t = (1, 2, 3)$, $R$ is a Nash equilibrium of the TDA mechanism. However, under the coarser tier structure $t^\prime = (1, 2, 2)$, $R$ is not a Nash equilibrium.\footnote{This is because under $R$, student 3 can deviate to $Q_3$: $b \mathrel{Q_3} c \mathrel{Q_3} 3 \mathrel{Q_3} a$, and be matched to school $b$.} Thus, $TDA(t)$ is not more manipulable than $TDA(t^\prime)$.$\hfill\square$
\end{example}

One intuition for the failure of this result is that the nested structure is for reshuffled Nash equilibria, not necessarily for Nash equilibria themselves. That is, if $R$ is a Nash equilibrium under the finer tier structure, then $\tilde R$ must be an equilibrium under a coarser tier structure, but $R$ itself might not be.

Next, we focus on stability. Although the TDA mechanism is not stable, it satisfies a weaker version of stability.\footnote{Weaker versions of stability have been studied by several works on matching under constraints \citep{abdulkadiroglu-2003, hafalir-2013, kamada-2015, kamada-2017}. But their analysis is independent of ours since their models do not feature a tier structure.}

\begin{definition}
    A matching $\mu$ is \textbf{stable with respect to the tier structure $t$} if 
    \begin{enumerate}
        \item $\mu$ is individually rational; and
        \item if $\mu$ has a blocking pair $(i, s)$, then $t_{\mu(i)} < t_s$.
    \end{enumerate}
\end{definition}

That is, we allow a blocking pair to exist if student $i$ prefers a school $s$ in a later tier than his current match $\mu(i)$. This concept reduces to the standard stability concept of \cite{gale-1962} if all schools are in the same tier.

A mechanism $f$ is stable with respect to the tier structure $t$ if, for any school choice problem $E$, $f(E)$ is stable with respect to the tier structure $t$.

\begin{proposition} \label{prop:stable-tier}
    The TDA mechanism is stable with respect to the tier structure $t$.\footnote{Note that this weaker stability condition allows the TDA mechanism to be wasteful. The problem of finding a non-wasteful sequential mechanism has been illustrated in \cite{dur-2018, andersson-2024, hatakeyama-2024}. See \Cref{sec:lit} for more details.}
\end{proposition}

Thus, one may intuitively think that the TDA mechanism becomes more stable as the tier structure becomes coarser. But \Cref{exp:stable} provides a counterexample.

\begin{example}[The TDA mechanism is not more stable as the tier structure becomes coarser] \label{exp:stable}
    Consider three students, $1, 2, 3$, and three schools, $a, b, c$, each with one seat. Student preferences $R$ as well as school priorities $\succsim$ are as follows:
    \begin{multicols}{2}
    \begin{center}
    \begin{tabular}{c|c|c}
        \textbf{$R_1$} & \textbf{$R_2$} & \textbf{$R_3$} \\
        \hline
        $c$ & $c$ & $b$ \\
        $a$ & $b$ & $c$ \\
        $b$ & $2$ & $3$ \\
        $1$ & $a$ & $a$ \\
    \end{tabular}
    \end{center}
   
    \columnbreak
    \begin{center}
    \begin{tabular}{c|c|c}
        \textbf{$\succsim_a$} & \textbf{$\succsim_b$} & \textbf{$\succsim_c$} \\
        \hline
        $1$ & $2$ & $3$ \\
        $2$ & $3$ & $1$ \\
        $3$ & $1$ & $2$ \\
    \end{tabular}
    \end{center}
    \end{multicols}

    Under the tier structure $t = (1, 2, 3)$, $TDA(t)(R) = ((1, a), (2, b), (3, c))$ is stable. However, under the coarser tier structure $t^\prime = (1, 2, 2)$, $TDA(t^\prime)(R) = ((1, a), (2, c), (3, b))$ is not stable - $(1, c)$ constitutes a blocking pair. Thus, $TDA(t^\prime)$ is not more stable than $TDA(t)$.$\hfill\square$
\end{example}

Let us explain where the intuition goes wrong. As suggested by \Cref{prop:stable-tier}, a blocking pair $(i,s)$ arises when student $i$ is forced to match with a school in an earlier tier, allowing potentially lower-ranked students to be matched with $s$. Actually, for the same reason, if the tier structure is refined, these lower-ranked students who were originally matched to $s$ may be forced to match with another school. In \Cref{exp:stable}, for instance, student 2 is rematched to school $b$ under $t$, causing the justified envy to disappear. Therefore, stability might be recovered after refining the tier structure.

\subsection{Undominated Nash Equilibrium} \label{sec:undominated}
In this section, we refine the solution concept to undominated Nash equilibrium. We show that the set of undominated Nash equilibrium outcomes under the TDA mechanism may contain more matchings than the one under the DA mechanism. This result provides a sanity check for the comparison between the DA and TDA mechanisms.

Strategy A is \textbf{weakly dominated} by strategy B if there is at least one set of opponents' actions for which A gives a worse outcome than B, while all other sets of opponents' actions give A the same payoff as B.
A Nash equilibrium is \textbf{undominated} if none of the students uses a weakly dominated strategy in this equilibrium. 
An undominated Nash equilibrium outcome is the outcome of an undominated Nash equilibrium.

First, we identify some weakly dominated strategies under the TDA mechanism to simplify the analysis follows. If a reported preference is inconsistent with the true preference within a tier, then it is either weakly dominated by a preference that is consistent within this tier or yields the same outcome as this consistent preference under all opponents' actions.

\begin{proposition} \label{prop:weakly}
    Given any school choice problem $E$ and tier structure $t$. 
    Fix any preference $Q_i$ such that there exist schools $s$ and $s^\prime$, with $t_s = t_{s^\prime}$, $s \mathrel{P_i} s^\prime$, and $s^\prime \mathrel{Q_i} s$. 
    Let $\hat S_{t_s} := \{s \in S_{t_s}: s \mathrel{Q_i} i\}$ be the set of tier $t_s$ schools that $Q_i$ ranks as acceptable. 
    Let preference $Q^*_i$ be such that $s_1 \mathrel{P_i} s_2 \iff s_1 \mathrel{Q^*_i} s_2$ for any $s_1, s_2 \in \hat S_{t_s}$, and $s_1 \mathrel{Q_i} s_2 \iff s_1 \mathrel{Q^*_i} s_2$ for any $s_1, s_2 \in (S \backslash \hat S_{t_s}) \cup \{i\}$
    
    Then $Q_i$ is either weakly dominated by $Q^*_i$, or $TDA(t)(Q_i, Q_{-i})(i) = TDA(t)(Q^*_i, Q_{-i})(i)$ for any $Q_{-i} \in \mathcal{R}^{|I|-1}$.
\end{proposition}

The proof for \Cref{prop:weakly} follows from the strategy-proofness of the DA mechanism. That is, for any opponents' actions $Q_{-i}$, $TDA(t)(Q^*_i, Q_{-i})(i) \mathrel{R_i} TDA(t)(Q_i, Q_{-i})(i)$.

One implication of \Cref{prop:weakly} is that if an inconsistent preference $Q_i$ is weakly undominated, then the corresponding consistent preference $Q^*_i$ is also weakly undominated, and both strategies yield the same outcome. Thus, when determining the set of undominated Nash equilibrium outcomes, we only need to consider preferences that are consistent with the true preferences within tiers.

Then, we state two features regarding the undominated Nash equilibrium outcomes. 
First, the undominated Nash equilibrium outcome of the TDA mechanism can be unstable.
Consider \Cref{exp:TDA}, the Nash equilibrium $Q$ is in fact undominated. 
For students 1 and 2, they only need to consider whether to list school $a$ as acceptable. If student 1 lists $a$ as acceptable, he will be strictly worse off given $Q_{-1}$. The same applies to student 2 if he does not list $a$ as acceptable. For student 3, listing $a$ as acceptable or listing $c$ as unacceptable both lead to a strictly worse outcome under either $Q_{-3}$ or the case where students 1 and 2 report only school $b$ as acceptable.

Second, some Nash equilibrium outcomes of the TDA mechanism are dominated.
Consider \Cref{exp:TDA} again, but change the tier structure to $t^\prime = (2, 1, 1)$. In this case, the matching $\mu = ((1, c), (2, a), (3, b))$ remains a TDA Nash equilibrium outcome under $t^\prime$. However, $\mu$ is dominated. This is because now the unique undominated strategy for student 2 is $R_2$, and the undominated strategies for student 1 are $Q_1$ and $R_1$. Then, no matter what student 3 reports, the TDA outcome cannot be $\mu$.

Let us consider an implication of the above analysis. 
For the DA mechanism, the unique undominated Nash equilibrium outcome is the SOSM \citep{kumano-2011, kumano-2012}. In contrast, the TDA mechanism's undominated Nash equilibrium outcomes are not necessarily stable and may result in strictly worse outcomes for top-tier schools, as shown in \Cref{exp:TDA}. Therefore, when we impose a sanity check by assuming that students do not play weakly dominated strategies and compare the undominated TDA equilibrium outcomes with the SOSM, the findings support our earlier conclusion that a transition from the DA to the TDA mechanism may not achieve the intended goal.

\subsection{Incomplete Information} \label{sec:incomplete}
In this section, we study the preference revelation game under the TDA mechanism when students have incomplete information about others' preferences or schools' priorities. We give two messages in this section: the Nash equilibrium analysis results (\Cref{thm:nested,thm:implementation}) do not extend, and the TDA mechanism still can fail the goal of guaranteeing that top-tier schools are better off. 

Here, we only cover the case where students are uncertain about school priorities; \Cref{exp:prefun} in \Cref{apx:B} discusses the scenario of preference uncertainty.

\begin{example}[Priority uncertainty] \label{exp:prioun}
    Consider three students, $1, 2, 3$, and three schools, $a, b, c$, each with one seat. School $a$ is in tier 1, and schools $b$ and $c$ are in tier 2. All students are expected utility maximizers. They share identical types (i.e., utility functions), $U_1, U_2, U_3$, and thus the same preferences, $R_1, R_2, R_3$:
    \begin{multicols}{2}
        \begin{center}
        \begin{tabular}{c|c|c|c|}
             \multicolumn{1}{c}{} & \multicolumn{1}{c}{$U_1$}  & \multicolumn{1}{c}{$U_2$} & \multicolumn{1}{c}{$U_3$} \\\cline{2-4}
              $a$ & $2$ & $2$ & $2$  \\\cline{2-4}
              $b$ & $3$ & $3$ & $3$ \\\cline{2-4}
              $c$ & $1$ & $1$ & $1$ \\\cline{2-4}
              $i$ & $0$ & $0$ & $0$ \\\cline{2-4}
        \end{tabular}
        \end{center}

        \columnbreak
        \begin{center}
        \begin{tabular}{c|c|c}
             \textbf{$R_1$} & \textbf{$R_2$} & \textbf{$R_3$}\\
            \hline
            $b$ & $b$ & $b$\\
            $a$ & $a$ & $a$\\
            $c$ & $c$ & $c$\\
            $1$ & $2$ & $3$\\
        \end{tabular}
        \end{center}
    \end{multicols}

    Students face two potential priority profiles, $\succsim$ and $\succsim^\prime$, with probabilities $1/5$ and $4/5$, respectively.
    \begin{multicols}{2}
        \begin{center}
        \begin{tabular}{c|c|c}
            \textbf{$\succsim_a$} & \textbf{$\succsim_b$} & \textbf{$\succsim_c$}\\
            \hline
            $1$ & $1$ & $1$\\
            $2$ & $2$ & $2$\\
            $3$ & $3$ & $3$\\
        \end{tabular}
        \end{center}

        \columnbreak
        \begin{center}
        \begin{tabular}{c|c|c}
            \textbf{$\succsim_a^\prime$} & \textbf{$\succsim_b^\prime$} & \textbf{$\succsim_c^\prime$}\\
            \hline
            $2$ & $2$ & $2$\\
            $3$ & $3$ & $3$\\
            $1$ & $1$ & $1$\\
        \end{tabular}
        \end{center}
    \end{multicols}

    Note that in this preference revelation game, each student's strategy space is still $\mathcal{R}$. That is, students' strategies cannot depend on the realization of the priority profile. We make this restriction to model the reality where each student can only report one preference ranking.
   
    In the dominant-strategy equilibrium of the DA mechanism, everyone reports truthfully, and the respective outcomes under two realizations are:
    $$DA(q, \succsim, R) = \begin{pmatrix}
    1 & 2 & 3 \\
    b & a & c \\
    \end{pmatrix},
    \text{ }
    DA(q, \succsim^\prime, R) =
    \begin{pmatrix}
    1 & 2 & 3 \\
    c & b & a \\
    \end{pmatrix}.$$

    Considering the preference revelation game induced by the TDA mechanism, one Nash equilibrium is shown below:
    \begin{center}
        \begin{tabular}{c|c|c}
            \textbf{$Q_1$} & \textbf{$Q_2$} & \textbf{$Q_3$} \\
            \hline
            $b$ & $b$ & $b$\\
            $c$ & $c$ & $a$\\
            $1$ & $2$ & $c$\\
            $a$ & $a$ & $3$\\
        \end{tabular}
    \end{center}
    The corresponding TDA outcomes under two realizations are:
    $$
    TDA(t)(q, \succsim, Q) = \begin{pmatrix}
    1 & 2 & 3 \\
    b & c & a \\
    \end{pmatrix},
    \text{ }
    TDA(t)(q, \succsim^\prime, Q) =
    \begin{pmatrix}
    1 & 2 & 3 \\
    c & b & a \\
    \end{pmatrix}.
    $$

    In fact, the two matchings we derived under the DA and TDA mechanisms are the unique Nash equilibrium outcomes, respectively.

    First, we show that \Cref{thm:nested,thm:implementation} do not extend to this setting. First, note that the unique TDA equilibrium outcome $TDA(t)(q, \succsim, Q)$ is not stable, as $(2, b)$ forms a blocking pair - a failure of part 1 of \Cref{thm:nested}. Second, the priority profiles in both realizations satisfy within-tier acyclicity - a failure of \Cref{thm:implementation}. Third, the Nash equilibrium outcome of the TDA mechanism is not an equilibrium outcome under the DA mechanism - a failure of part 2 of \Cref{thm:nested}.

    Second, we show that the TDA mechanism may not improve the quality of students admitted by top-tier schools as intended.
    Notice that under one realization, $DA(q, \succsim^\prime, R)$ and $TDA(t)(q, \succsim^\prime, Q)$ are the same. Under another, however, the top-tier school $a$ ends up with the lower-priority student $3$ instead of student $2$, who was initially achievable under the DA mechanism.$\hfill\square$
\end{example}

This scenario closely reflects the reality of the current Chinese admission system. After the recent 2020 reform, students must report their preferences for some top-tier schools before taking the national examination, introducing uncertainty about their exam scores and thus, their priorities.\footnote{Currently, top-tier admissions are conducted on majors instead of schools. For consistency, we stick to the term ``schools'' when referring to the majors in this additional tier.}\footnote{As most top-tier schools assign a heavy weight, i.e., 85\% or more, to the national examination when ranking students, there tends to be an alignment in priority profile as shown in \Cref{exp:prioun}. We include \Cref{exp:prioun2} in \Cref{apx:B} where school priorities do not align.}

Media often justify this early reporting procedure by suggesting that top-tier schools can attract students who strongly prefer them, as these students are willing to risk early admission. Another rationale is that these schools might get high-scoring students who use the top-tier program as a ``safety net.''\footnote{Feng Wang, \textit{2023 Programs Explore Major Reforms: Multiple Schools Conduct Early Written Exam to Select Students with Specialized Talents}, $21^{\text{st}}$ Century Business Herald, April 27, 2023, \href{https://www.21jingji.com/article/20230427/herald/2d3350fbb5f0a66df984ee0da2f79641.html}{https://www.21jingji.com/article/20230427/herald/2d3350fbb5f0a66df984ee0da2f79641.html}, accessed 08/14/2024.}

However, \Cref{exp:prioun} refutes these points. Given any realization, since the priority is the same for all schools, the stable matching is unique. Thus, both the TDA and DA mechanisms would result in the SOSM under equilibrium. Yet, introducing uncertainty alters the outcome. Under realization $\succsim$, school $a$ ends up with a lower-priority (lower-scoring) student $3$ instead of the originally achievable student $2$, despite both students having the same preference for school $a$, and neither considers $a$ their top choice. This discrepancy arises because student $2$, expecting to score well, opts not to apply to the “safety net” school $a$, which also leads to a downgrade to a less preferred school $c$ when he underperforms.

\section{Related Literature} \label{sec:lit}
This paper is closely related to studies in sequential mechanisms. In such mechanisms, objects are allocated to agents in multiple stages (tiers), and agents who are assigned in earlier stages are excluded from participating in subsequent ones.

\cite{abdulkadiroglu-2005, abdulkadiroglu-2009} study the high school admission system in New York City. They express the concern that a two-stage system can create unstable outcomes. 
\cite{westkamp-2012} studies the German college admissions system, which uses the Boston mechanism in the first stage and the school-proposing DA mechanism in the second. He characterizes the set of Nash equilibrium outcomes of the mechanism as the stable set.

\cite{dur-2018} study a general two-stage sequential mechanism that might consist of different mechanisms in each stage. Two applications of their analysis are the Turkish state school teacher appointment system and the US public school admission system.\footnote{In Turkey, teachers are first assigned to tenured positions and then to contractual positions in state schools through a two-stage SD mechanism. See \cite{dur-2018} for more details of this mechanism.}\footnote{In New York City, the first tier contains exam schools, while the second tier includes regular high schools. Exam schools rank students using a uniform test, while regular schools rank students based on demographic criteria. Unlike the TDA mechanism, students can initially apply to both tier 1 and tier 2 schools simultaneously and choose between them if admitted by schools in both tiers. After this stage, unmatched students apply again to tier 2 schools to determine the final outcome. This variation is not modeled in this paper. See \cite{abdulkadiroglu-2005, abdulkadiroglu-2009} for an earlier account of the New York City mechanism. We thank Al Roth for pointing out the distinction between the New York City admission system and the TDA mechanism.}
They show that no two-stage mechanism comprising the DA, Boston, and top trading cycle (TTC) mechanisms can be straightforward, meaning that some students can gain by misreporting in some stages they participate in.
They also conduct Nash equilibrium analysis and show that all stable matchings are TDA Nash equilibrium outcomes. Our part 1 of \Cref{thm:nested} generalizes this result to the case with multiple tiers.

Motivated by the school admission systems in Turkey and Sweden, \cite{andersson-2024} model the sequential assignment process as an extensive form game, by assuming students can observe the moves in former tiers before reporting their preferences for the next tier.\footnote{In the Greater Stockholm Region, Sweden, students are first assigned to private primary schools, and then unmatched students are assigned to public primary schools. The priority in private schools is given on a first-come-first-served basis, while the priority for public schools is decided based on relative distance to schools. See \cite{andersson-2017} and \cite{andersson-2024} for more details of this mechanism.}
Accordingly, they use the solution concept of subgame perfect Nash equilibrium (SPNE).\footnote{In China, students submit their preferences twice: first for schools in tiers 1 and 2, and again for tier 3 schools if they are still unmatched. Some provinces announce which schools have available seats at the end of each round, allowing unmatched students to resubmit their preferences for the schools in that round. This extra stage is not modeled in this paper. 
In Turkey, students initially apply to private schools in a decentralized manner. Those who remain unmatched then submit their preferences for public schools.
For simplicity, we model these systems as simultaneous games rather than extensive form games, using the solution concept of Nash equilibrium rather than SPNE. Moreover, \posscite{andersson-2024} assumption that students' moves are observable is not satisfied in the Chinese setting.}
For any two-stage TDA mechanism, they show that there exists a problem under which the unique SPNE is not straightforward. They also show that if the priority structure is acyclic, a two-stage TDA mechanism implements stable matchings in equilibrium. Under the solution concept of Nash equilibrium, our \Cref{thm:implementation} generalizes their implementability result to allow for cycles across tiers and additional stages.

\cite{hatakeyama-2024} studies the problem of designing a tier structure under a multi-stage mechanism to obtain straightforwardness and non-wastefulness. One application of his analysis is the employment exams for public officers.\footnote{In Japan, exams for the most prestigious national employee positions are held first, followed by those for positions at the prefecture and city levels. For more details, see \href{https://90r.jp/schedule.html}{https://90r.jp/schedule.html}, accessed 08/15/2024.}
He shows that when the preference profile is tiered, for any multi-stage mechanism comprising the DA, TTC, and SD mechanisms, there exists a tier structure where the mechanism is straightforward and non-wasteful.\footnote{\label{ft:tiered-pref}A preference profile is tiered if there exists a tier structure such that the preference profile aligns with the tier structure.} Our \Cref{prop:st-pf} can be seen as an application of this result, by specifying the tier structure under the TDA mechanism.

To summarize, our work provides three new results regarding the TDA mechanism. First, we allow for the case of multiple tiers. This allows us to consider the impact of refining the tier structure and reveals the nest structure of the TDA mechanism (\Cref{thm:nested}). Second, we weaken the acyclicity condition to within-tier acyclicity, suggesting that cycles across tiers do not generate unstable outcomes in equilibrium (\Cref{thm:implementation}). Third, we study the welfare of schools under the TDA mechanism (\Cref{thm:unique}), while other studies focus on properties such as straightforwardness and non-wastefulness.

Complementary to the analysis of sequential mechanisms, several works investigate parallel mechanisms. In such mechanisms, students have the option to apply to subsequent stages of schools even after getting assigned to earlier ones. However, if they accept an offer from a later stage, they must give up their previously allocated seat.

Among those studies, several papers focus on the parallel use of the DA mechanism. 
\cite{manjunath-2016} show that the outcomes under such a parallel mechanism can be inefficient, and propose an iterative mechanism to rematch students and improve welfare. \cite{turhan-2019} and \cite{afacan-2022} further study the properties of this proposed mechanism. \cite{haeringer-2021} analyze a multi-stage college admission mechanism in France, introducing refitting rules to improve student welfare. \cite{dogan-2019, dogan-2023} analyze a two-stage school admission mechanism in Chicago, examining the welfare effects of adding one more tier. Instead of focusing on the welfare of students, \cite{ekmekci-2019} focus on the incentives for schools to join a centralized admission system versus staying in a parallel one, finding centralized admissions not incentive-compatible for all schools to join in general.

Our work differs from the analysis of parallel mechanisms in two aspects. 
First, studies in parallel mechanisms worry about wastefulness, since students dropping from matched schools in earlier tiers generate empty seats. However, as illustrated in \Cref{sec:implementation}, all TDA equilibrium outcomes are non-wasteful.
Second, the two mechanisms bear different welfare properties, both for schools (\Cref{rmk:school-incentive}) and for students (\Cref{rmk:student-welfare}).

\begin{remark} \label{rmk:school-incentive}
    Let us consider a scenario where schools can also strategize. Specifically, we consider that when all other schools are in one tier, whether a school has an incentive to join. This scenario captures a decentralized admission process where schools can unilaterally evade the central admission system. 
    \cite{ekmekci-2019} show that in a two-stage parallel DA mechanism, every school weakly prefers to be in tier 2 if all other schools are in tier 1.
    This result does not apply to our setting: under the TDA mechanism, schools may not have an incentive to unilaterally evade. In \Cref{exp:TDA}, it can be checked that under both tier structures $t = (1, 2, 2)$ and $t^\prime = (2, 1, 1)$, the set of TDA equilibrium outcomes remains the same, and school $a$ ends up with a worse equilibrium outcome than under the SOSM.$\hfill\square$
\end{remark}

\begin{remark} \label{rmk:student-welfare}
    The TDA mechanism has ambiguous welfare effects on students. The TDA mechanism is not Pareto efficient, and not all TDA equilibrium outcomes Pareto dominate the SOSM - as shown in \Cref{exp:TDA}. \cite{dogan-2023} show that, under a two-stage parallel DA mechanism, when students can only play truncation strategies, an additional stage improves the welfare of all students. This result does not apply to our setting: in \Cref{exp:TDA}, everyone is playing a truncation strategy in $Q$, but student 2 is worse off than under the SOSM.$\hfill\square$
\end{remark}

\begin{remark}\label{rmk:FDA}
    Our paper also relates to the study of distributional constraints, where schools are partitioned into different regions and restrictions are placed on the total number of students assigned to each region (``regional caps'').
    In practice, distributional constraints have been used to achieve goals similar to those discussed in \Cref{sec:goal}. For example, in China, there are two types of master's degrees: academic and professional. In 2010, to increase the number of students pursuing professional master's degrees, the Chinese government imposed a cap on the number of academic master's students \citep{kamada-2015}.
    
    \cite{kamada-2015} study medical residency matching in Japan, which features such regional caps, and propose the flexible deferred acceptance (FDA) mechanism. This mechanism is strategy-proof, so we focus only on its outcome under truthful reporting.
    They show that when a regional cap is not binding, schools within that region are weakly better off under the FDA mechanism compared to the DA mechanism (without regional caps).
    Thus, by treating each tier as a region and setting the regional cap for top-tier schools above their combined quotas, we can guarantee that top-tier schools are weakly better off.$\hfill\square$
\end{remark}

\section{Conclusion} \label{sec:conclude}
In this paper, we analyze the Nash equilibria under the TDA mechanism, showing the nested structure of equilibrium outcomes and identifying a necessary and sufficient condition for implementing stable matchings. Our results show that simply dividing schools into tiers does not necessarily enhance the quality of students matched to top-tier schools.

One future direction involves analyzing scenarios where students are partially sophisticated. While comparative statics analysis by \cite{roth-1990} suggests that schools in the top tier benefit when all students are fully naive, i.e., reporting truthfully, our research presents a contrasting result where all students are fully sophisticated. Thus, a question arises: how much rationality is needed for the TDA mechanism to fail to achieve its objective? One possible approach is to adopt the framework from \cite{pathak-2008}, assuming a portion of students as fully naive and another as fully sophisticated, and analyze whether there exists a cutoff in the proportion that affects the mechanism's outcome.

Another possible direction is the equilibrium analysis in large markets. As the market size increases, profitable deviation becomes more difficult \citep{azevedo-2016}. Therefore, whether the TDA mechanism is strategy-proof in the large \citep{azevedo-2018}, and whether the Nash equilibrium outcomes can enable the mechanism to achieve its goal remains to be explored.

\newpage
\appendix
\begin{center}
    \textbf{\LARGE Appendix}
\end{center}
\section{Proofs} \label{apx:A}
We begin by citing two useful lemmas to support the proofs below.

\begin{lemma}[\citealt{roth-1990}, Lemma 4.8]\label{lem:topchoice}
    Given any school choice problem $E$, let $Q_i$ be a preference list where the first choice is $DA(R)(i)$ if $DA(R)(i) \neq i$, and the empty list otherwise. Then, $DA(Q_i, R_{-i})(i) = DA(R)(i)$.
\end{lemma}

\begin{lemma}[\citealt{roth-1990}, Theorem 5.35]\label{lem:lessstudent}
    Consider two school choice problems $E=(I, S, q, \succsim, R)$ and $E^\prime = (I^\prime, S, q, (\succsim_i)_{i \in I^\prime}, (R_i)_{i \in I^\prime})$, where $I^\prime \subseteq I$. Let $\mu_1$ and $\mu_2$ denote the SOSMs under $E$ and $E^\prime$ respectively. Then, for any school $s \in S$, $\mu_1^{-1}(s) \mathrel{\underbar{\text{$\triangleright$}}_s} \mu_2^{-1}(s)$, and for any student $i \in I^\prime$, $\mu_2(i) \mathrel{R_i} \mu_1(i)$. Symmetrical results are obtained if $S^\prime$ is contained in $S$.
\end{lemma}

\subsection{Proof of Theorem \ref{thm:nested}}
\subsubsection*{Proof of Part 1}
This proof follows a similar argument as the proof of Theorem 1 in \cite{ergin-2006}.

Consider any school choice problem $E$ and tier structure $t$, and let $\mu$ be a stable matching with respect to the preferences $R$ in $E$. Define a preference profile $Q = (Q_i)_{i \in I}$ where each student $i$ lists $\mu(i)$ as his only acceptable choice in his stated preference $Q_i$. Under $Q$, each round of the TDA mechanism terminates at the first step, and each student is assigned a seat at his first choice based on the stated preferences. Hence, $TDA(t)(Q) = \mu$.

Next, we show that $Q$ is a Nash equilibrium under the TDA mechanism in two steps. Suppose towards a contradiction that there exists a student $i$ and an alternate preference $Q^\prime_i$ such that $TDA(t)(Q^\prime_i, Q_{-i})(i) \mathrel{P_i} \mu(i)$. Let us denote $TDA(t)(Q^\prime_i, Q_{-i})(i)$ as $s$. 
First, when $i$ reports $Q^\prime_i$, it must be the case that $i$ has not been temporarily accepted by any school before proposing to $s$. This is because, if $i$ were temporarily accepted by any school $s^\prime$ earlier than $s$, then one of two cases would apply: either $|\mu^{-1}(s^\prime)| < q_{s^\prime}$, or $s^\prime$ rejects another student $j$ who then stays unmatched. Under both cases, $i$ will be matched to $s^\prime$ instead of $s$ - a contradiction. Thus, when $i$ proposes to $s$, the current assignments for all other students remain unchanged from the one under $TDA(t)(Q)$. Second, since $i$ is accepted by $s$, two situations arise: either $|\mu^{-1}(s)| < q_{s}$, or there exists a student $j \in \mu^{-1}(s)$ such that $i \succ_s j$. But under both cases, $\mu$ is not stable - a contradiction. Therefore, no alternative preference $Q^\prime_i$ can improve $i$'s outcome over $\mu(i)$. Hence, $Q$ is a TDA Nash equilibrium, and $\mu = TDA(t)(Q)$ is a Nash equilibrium outcome under the TDA mechanism.

\subsubsection*{Proof of Part 2}
\subsubsection*{Step 1: Reshuffling Preferences}
We begin by introducing a mapping $(Q, t) \mapsto \widetilde Q$, such that for any tier structure $t$, the result of the mapped preference profile $\widetilde Q$ under the DA mechanism is equal to the result of the original profile $Q$ under the TDA mechanism, i.e., $DA(\widetilde Q) = TDA(t)(Q)$. 

\begin{definition}
    Given any tier structure $t$, for any student $i$ and preference $Q_i$, the \textbf{reshuffled preference} $\widetilde{Q}_i$ is constructed iteratively:

    Step 1: Take $i$'s acceptable schools in tier 1, rank them in the same order as in $Q_i$, and then place these schools at the first position in $\widetilde Q_i$.

    In general,
   
    Step $k$: Take $i$'s acceptable schools in tier $k$, rank them in the same order as in $Q_i$, and then place these schools at the $k$-th position in $\widetilde Q_i$.
   
    This process continues until there are no acceptable schools left for $i$, at which point rank all remaining schools behind $i$ in $\widetilde Q_i$.\footnote{For instance, in \Cref{exp:TDA}, student 1's reshuffled true preference $\widetilde R_1: a - c - b - 1$, and his reshuffled Nash equilibrium strategy $\widetilde Q_1: c - b - 1 - a$.}
\end{definition}

With a slight abuse of notation, for a preference profile $Q:= (Q_i)_{i \in I}$, we shorthand $\widetilde Q:= (\widetilde Q_i)_{i \in I}$.

\begin{remark} \label{rmk:reshuffled}
    We further explain the reshuffled preference in this remark. 
    Another way to model the TDA mechanism is by having students report their preferences for schools only in the next tier at the start of each round (sequential reporting), instead of listing their preferences for all schools at the beginning (one-shot reporting).
    If no new information is given after each round, then the two models are the same, as discussed in \cite{dur-2018} footnote 22. Then, the reported preferences from the sequential reporting are reshuffled preferences from one-shot reporting.$\hfill\square$
\end{remark}

\begin{lemma}\label{lem:DA&TDA}
    For any preference profile $Q$ and tier structure $t$, the outcome of the TDA mechanism is equal to the outcomes of the TDA and DA mechanisms with reshuffled preference profile, i.e., $$TDA(t)(Q) = TDA(t)(\widetilde Q) = DA(\widetilde Q).$$
\end{lemma}
\begin{proof}
    The proof of the first equality is that, in each round of the TDA mechanism, the outcome is affected only by the relative rankings of schools within the same tier, and reshuffling the preferences preserves these intra-tier rankings.

    For the second equality, consider a variation of the DA mechanism under any tier structure and its corresponding reshuffled preference profile $\widetilde Q$.
    \begin{itemize}
        \item In the first stage, every student $i$ proposes to schools in the first tier according to $\widetilde Q_i$. If $i$ has been rejected by all acceptable tier 1 schools in $\widetilde Q_i$, he stops proposing. This stage ends when no more proposals to tier 1 schools are rejected.
        \item In the $k$-th stage, every student $i$ who is unmatched before stage $k$ continues proposing to schools in tier $k$ according to $\widetilde Q_i$. If student $i$ has been rejected by all acceptable tier $k$ schools in $\widetilde Q_i$, he stops proposing. This stage ends when no more proposals to tier $k$ schools are rejected.
        \item The similar process continues until the last tier, then the unmatched students are assigned to themselves.
    \end{itemize}

    Next, we explain that this procedure generates the same matching as in the original DA mechanism. At the beginning of the second stage, every proposing student $i$ has already proposed to all of their acceptable schools in tier 1 under $\widetilde Q_i$, thus none of them will propose to tier 1 schools again. Since these schools never get new proposals after stage 1, students who are matched to tier 1 schools at the end of stage 1 remain unchanged. The same argument applies to the rest stages. Hence, this variation only changes the timing of making proposals, which makes it an equivalent version of the DA mechanism.\footnote{As the proof of Theorem 2 in \cite{gale-1962} does not rely on when students make their proposals, we know this variation also produces the SOSM, which is the same matching produced by the original DA mechanism.}

    The last thing to notice is that the matching at each stage of this DA mechanism is the same as the one in each round of the TDA mechanism. To conduct a proof by induction, we see that the first stage of the DA mechanism matches exactly with the first round of the TDA mechanism. Since the matches with tier 1 schools stay the same in later rounds, we can effectively remove these schools and apply the same base case logic to the reduced set of schools. This step-by-step equivalence completes the proof.
\end{proof}

\subsubsection*{Step 2: Establishing Nash Equilibrium Relation}
In Steps 2 and 3 of the proof of part 2 of \Cref{thm:nested}, to distinguish the reshuffling operation under these two tier structures, we denote $(Q, t) \mapsto \widetilde Q$ as the reshuffling with respect to the finer tier structure $t$, and $(Q, t^\prime) \mapsto \bar Q$ with respect to the coarser tier structure $t^\prime$.

\begin{lemma}\label{lem:tech}
    Given any school choice problem $E$ and tier structure $t$, take any TDA Nash equilibrium $Q$, for any student $i$, the outcome for the reshuffled preference profile under the DA mechanism is equal to the outcome when $i$ reports truthfully, i.e., $$DA(\widetilde{Q}_i, \widetilde{Q}_{-i})(i) = DA(R_i, \widetilde{Q}_{-i})(i).$$
\end{lemma}

\begin{proof}
    Consider any school choice problem $E$ and tier structure $t$, and take any TDA Nash equilibrium $Q \in \mathcal{E}^{TDA(t)}(E)$. By definition, for any student $i$ and any preference $Q^\prime_i$, $TDA(t)(Q_i, Q_{-i})(i) \mathrel{R_i} TDA(t)(Q_i^\prime, Q_{-i})(i)$. Replacing each term using \Cref{lem:DA&TDA}, we have $DA(\widetilde Q_i,\widetilde Q_{-i})(i) \mathrel{R_i} DA(\widetilde Q_i^\prime,\widetilde Q_{-i})(i)$.

    Let $Q^\prime_i$ be such that the only acceptable school is $DA(R_i, \widetilde Q_{-i})(i)$, then by \Cref{lem:topchoice}, $DA(Q_i^\prime, \widetilde Q_{-i})(i) = DA(R_i, \widetilde Q_{-i})(i)$. Since $\widetilde Q_i^\prime = Q_i^\prime$, we have $DA(\widetilde Q_i^\prime, \widetilde Q_{-i})(i) = DA(R_i, \widetilde Q_{-i})(i)$, and thus $DA(\widetilde Q_i,\widetilde Q_{-i})(i) \mathrel{R_i} DA(R_i, \widetilde Q_{-i})(i)$.

    Finally, as the DA mechanism is strategy-proof, we have $DA(R_i, \widetilde Q_{-i})(i) \mathrel{R_i} DA(\widetilde Q_i,\widetilde Q_{-i})$. Hence, we recover $DA(\widetilde Q_i, \widetilde Q_{-i})(i) = DA(R_i, \widetilde Q_{-i})(i)$ for any $i \in I$.
\end{proof}

\begin{lemma}\label{lem:NErelation}
    Given any school choice problem $E$ and two tier structures $t$ and $t^\prime$, if $t$ is a refinement of $t^\prime$, then take any TDA Nash equilibrium $Q$ under $t$, the corresponding reshuffled preference profile $\widetilde Q$ is a TDA Nash equilibrium under $t^\prime$, i.e., $$Q \in \mathcal {E}^{TDA(t)}(E) \Rightarrow \widetilde Q \in \mathcal E^{TDA(t^\prime)}(E).$$
\end{lemma}
\begin{proof}
    Consider the school choice problem $E$ and two tier structures $t$ and $t^\prime$, where $t$ is a refinement of $t^\prime$. For any $Q \in \mathcal{E}^{TDA(t)}(E)$, by definition, we have for any student $i \in I$ and any preference $Q^\prime_i$, $TDA(t)(Q_i, Q_{-i})(i) \mathrel{R_i} TDA(t)(Q_i^\prime, Q_{-i})(i)$. By \Cref{lem:DA&TDA}, we have $DA(\widetilde Q_i, \widetilde Q_{-i})(i) \mathrel{R_i} DA(\widetilde Q_i^\prime, \widetilde Q_{-i})(i)$.
   
    By \Cref{lem:tech}, we know $DA(\widetilde Q_i, \widetilde Q_{-i})(i) = DA(R_i, \widetilde Q_{-i})(i)$. Combining with the DA mechanism being strategy-proof, we know for any preference $Q^\prime_i$, $DA(\widetilde Q_i, \widetilde Q_{-i})(i) \mathrel{R_i} DA(\bar Q^\prime_i, \widetilde Q_{-i})(i)$.
   
    Then, we use a series of equalities to get the final result. First, notice for any student $i \in I$ and preference $Q^\prime_i$, $TDA(t^\prime)(\widetilde Q_i, \widetilde Q_{-i})(i) = DA(\bar{\widetilde Q}_i, \bar{\widetilde Q}_{-i})(i) = DA(\widetilde Q_i, \widetilde Q_{-i})(i)$, where the first equality comes from \Cref{lem:DA&TDA} and the second one comes from the fact that $ t$ is a refinement of $t^\prime$. Second, using the same argument, we know that for any student $i \in I$ and preference $Q^\prime_i$, $TDA(t^\prime)(Q_i^\prime, \widetilde Q_{-i})(i) = DA(\bar Q_i, \bar{\widetilde Q}_{-i})(i) = DA(\bar Q_i^\prime, \widetilde Q_{-i})(i)$.
   
    Thus, we have $TDA(t^\prime)(\widetilde Q_i, \widetilde Q_{-i})(i) \mathrel{R_i} TDA(t^\prime)(Q_i^\prime, \widetilde Q_{-i})(i)$ holds for any student $i$ and preference $Q^\prime_i$. Hence, $\widetilde Q$ is a TDA Nash equilibrium under the coarser tier structure.
\end{proof}

\subsubsection*{Step 3: Concluding the Nash Equilibrium Outcome Relation}
Finally, we are ready to prove part 2 of \Cref{thm:nested}.

Given any school choice problem $E$ and two tier structures $t$ and $t^\prime$, such that $ t$ is a refinement of $t^\prime$. Take any matching $\mu \in \mathcal O^{TDA(t)}(E)$, then by definition, there exists $Q \in \mathcal E^{TDA(t)}(E)$ such that $\mu = TDA(t)(Q)$.

By \Cref{lem:NErelation}, $\widetilde Q \in \mathcal{E}^{TDA(t^\prime)}(E)$, and thus $TDA(t^\prime)(\widetilde Q) \in \mathcal{O}^{TDA(t^\prime)}(E)$. By \Cref{lem:DA&TDA}, $TDA(t)(Q) = DA(\widetilde Q)$ and $TDA(t^\prime)(\widetilde Q) = DA(\bar{\widetilde Q})$. Since $ t$ is a refinement of $t^\prime$, $DA(\bar{\widetilde Q}) = DA({\widetilde Q})$. Thus, $\mu = TDA(t)(Q) = TDA(t^\prime)(\widetilde Q) \in \mathcal{O}^{TDA(t^\prime)}(E)$.
   
\subsection{Proof of Theorem \ref{thm:implementation}}
We first define a non-bossy mechanism and show three related lemmas.

\begin{definition}[\citealt{satterthwaite-1981}]
    A mechanism $f$ is \textbf{non-bossy} if for any school choice problem $E = (q, \succsim, R)$, for any student $i \in I$ and preference $Q_i \in \mathcal{R}$, $f(q, \succsim, R)(i) = f(q, \succsim, Q_i, R_{-i})(i)$ implies $f(q, \succsim, R) = f(q, \succsim, Q_i, R_{-i})$.
\end{definition}

\begin{lemma}[\citealt{ergin-2002}, Theorem 1] \label{lem:nonbossy}
    For any school choice problem $E$, if the priority structure $(q, \succsim)$ is acyclic, then the DA mechanism is non-bossy.
\end{lemma}

\begin{lemma}[\citealt{haeringer-2009}, Lemma A.2] \label{lem:IR&NW}
    For any school choice problem $E$, all Nash equilibrium outcomes under the DA mechanism are individually rational and non-wasteful.
\end{lemma}
   
\begin{lemma} \label{lemma:kcycle}
    For any school choice problem $E$ and tier structure $t$, take any tier $k$, if there are no cycles within tier $k$, then for any $\mu \in \mathcal O^{TDA(t)}(E)$, $\mu$ does not have any blocking pair $(i, s)$ such that school $s$ is in tier $k$.
\end{lemma}
\begin{proof}    
    Suppose there exist such $s \in S_k$ and $\mu \in \mathcal O^{TDA(t)}(E)$ such that $(i, s)$ is a blocking pair under $\mu$. By definition there exists $Q \in \mathcal E^{TDA(t)}(E)$ such that $TDA(t)(Q) = \mu$.

    Suppose that student $i$ reports the empty list $Q_i^\prime$ where he lists no school as acceptable, then $TDA(t)(Q_i^\prime, Q_{-i})(i) = i$. By \Cref{lem:DA&TDA}, $TDA(t)(Q) = DA(\widetilde Q)$ and $TDA(t)(Q_i^\prime, Q_{-i}) = DA(\widetilde{Q}_i^\prime,\widetilde{Q}_{-i})$. By \Cref{lem:lessstudent}, $
    DA(\widetilde Q)^{-1}(s) \mathrel{\underbar{\text{$\triangleright$}}_s}  DA(\widetilde{Q}_i^\prime,\widetilde{Q}_{-i})^{-1}(s)
    $. Thus, $i$ still has higher priority than someone (or the empty seat) that is finally matched to $s$ under $TDA(t)(Q_i^\prime, Q_{-i})$.

    Suppose $i$ reports $Q_i^*$ which lists $s$ as the only acceptable school, then $TDA(t)(Q_i^*, Q_{-i})(i) = i$ as $Q$ is a Nash equilibrium strategy. Then, before the TDA mechanism reaches $s$'s tier, all steps are the same as under $TDA(t)(Q_i^\prime, Q_{-i})$. When the TDA mechanism reaches $s$'s tier, since this round of DA is non-bossy by \Cref{lem:nonbossy} and $i$ remains unmatched, the matches of others stay the same as under $TDA(t)(Q_i^\prime, Q_{-i})$ at the end of this round. After $s$'s tier, all steps remain the same. Thus, $TDA(t)(Q_i^\prime, Q_{-i}) = TDA(t)(Q_i^*, Q_{-i})$.

    As we have analyzed, under $TDA(t)(Q_i^*, Q_{-i})$, $i$ still has higher priority than someone (or the empty seat) that is finally matched to $s$. Thus, $TDA(t)(Q_i^*, Q_{-i})$ is not stable under $(Q_i^*, Q_{-i})$ - a contradiction.
\end{proof}

Then, we are ready to prove \Cref{thm:implementation}.

For the ``only if'' direction, suppose towards a contradiction that there exist a school choice problem $E$ and tier structure $t$ where $(q, \succsim, t)$ is within-tier acyclic, and there exists an unstable $\mu \in \mathcal{O}^{TDA(t)}(E)$. Following from \Cref{thm:nested} and \Cref{lem:IR&NW}, $\mu$ has at least one blocking pair $(i, s)$. By \Cref{lemma:kcycle}, we know that school $s$ does not belong to any tier - a contradiction.

For the ``if'' direction, we prove it by contrapositive, using a similar argument as the proof of Theorem 1 in \cite{ergin-2002}. Let $(I, S, q, \succsim, t)$ be given. Suppose $(q, \succsim, t)$ is not within-tier acyclic, i.e., there exist cycles within some tiers. Take any of such cycles, which involves $a, b, i, j, k, I_a$ and $I_b$ as in \Cref{def:acyclicity}, and denote the corresponding tier by $k$. Consider a preference profile $R$ where students in $I_a$ and $I_b$ respectively rank $a$ and $b$ as their top choice, and the preferences of $i, j$ and $k$ are such that $b \mathrel{P_i} a \mathrel{P_i} i$, $a \mathrel{P_j} j$ and $a \mathrel{P_k} b \mathrel{P_k} k$.\footnote{The preferences behind being unmatched can be arbitrary. We omit them here for simplicity.} Finally, let students outside $I_a \cup I_b \cup \{i, j, k\}$ prefer not to be assigned to any school. Then, the TDA outcome $\mu$ for $R$ is such that $\forall l \in I_a \cup \{i\}: \mu(l) = a$, $\forall l \in I_b \cup \{k\}: \mu(l) = b$ and $\forall l \in I \setminus (I_a \cup I_b \cup \{i, k\}): \mu(l) = l$.

Consider a Nash equilibrium preference profile $Q$ where all students except $j$ report $R$, and $j$ reports prefer not to be assigned to any school. Then, the corresponding TDA Nash equilibrium outcome $\mu^\prime$ for $Q$ is such that $\forall l \in I_a \cup \{k\}: \mu^\prime(l) = a$, $\forall l \in I_b \cup \{i\}: \mu^\prime(l) = b$ and $\forall l \in I \setminus (I_a \cup I_b \cup \{i, k\}): \mu^\prime(l) = l$. Now, notice that $\mu^\prime$ is not stable with respect to $R$, as $(j, a)$ forms a blocking pair.

\subsection{Proof of Theorem \ref{thm:unique}}
The ``if'' direction directly follows from \Cref{thm:implementation}: for any school choice problem $E$ and tier structure $t$, since $\mathcal{O}^{TDA(t)}(E) = \mathcal{S}(E)$, we know that all schools prefer any stable matching over the SOSM.

The ``only if'' direction is proven by contrapositive. Fix any school $s^* \in S^*$. 

Suppose there is another school $s^\prime$ within $s^*$'s tier, then we can construct a $(q, \succsim, R)$ that makes $s^*$ worse off.
Take any three students $i, j, k \in I$. Let $q_s = 1$ for any $s \in S$. Let $\succsim_{s^*}$ be $i \mathrel{\succsim_{s^*}} j \mathrel{\succsim_{s^*}} k \mathrel{\succsim_{s^*}} l$, where $l \in I \backslash \{i, j, k\}$. Let $\succsim_{s^\prime}$ be $k \mathrel{\succsim_{s^\prime}} i \mathrel{\succsim_{s^\prime}} l$, where $l \in I \backslash \{i, k\}$. 
Let $R_i$ be $s^\prime \mathrel{P_i} s^* \mathrel{P_i} i$; $R_j$ be $s^* \mathrel{P_j} j$; $R_k$ be $s^* \mathrel{P_k} s^\prime \mathrel{P_k} k$; $R_l$ be such that no school is acceptable to $l$.\footnote{The preferences behind being unmatched can be arbitrary. We omit them here for simplicity. The same applies to the preferences in the next paragraph.} Then the SOSM is $((i, s^*), (k, s^\prime))$. 
One Nash equilibrium $Q$ is such that $Q_i = R_i$, $Q_k = R_k$, $Q_l = R_l$, $Q_j$ is such that no schools is acceptable. Then $TDA(t)(Q) = ((i, s^\prime), (k, s^*))$. This makes $s^*$ strictly worse off.

Suppose there are two schools $s_1$ and $s_2$ in a tier other than $s^*$'s tier, then we can construct a $(q, \succsim, R)$ that makes $s^*$ worse off.
Take any three students $i, j, k \in I$. Let $q_s = 1$ for any $s \in S$. Let $\succsim_{s^*}$ be $k \mathrel{\succsim_{s^*}} j \mathrel{\succsim_{s^*}} l$, where $l \in I \backslash \{j, k\}$. Let $\succsim_{s_1}$ be $i \mathrel{\succsim_{s_1}} j \mathrel{\succsim_{s_1}} k \mathrel{\succsim_{s_1}} l$, where $l \in I \backslash \{i, j, k\}$. Let $\succsim_{s_2}$ be $k \mathrel{\succsim_{s_2}} i \mathrel{\succsim_{s_2}} l$, where $l \in I \backslash \{i, k\}$.
Let $R_i$ be $s_2 \mathrel{P_i} s^* \mathrel{P_i} s_1 \mathrel{P_i} i$; $R_j$ be $s_1 \mathrel{P_j} s^* \mathrel{P_j} s_2 \mathrel{P_j} j$; $R_k$ be $s_1 \mathrel{P_k} s^* \mathrel{P_k} s_2 \mathrel{P_k} k$; $R_l$ be such that no school is acceptable to $l$. Then the SOSM is $((i, s_2), (j, s_1), (k, s^*))$.
One Nash equilibrium $Q$ is such that $Q_i: s_2 \mathrel{Q_i} s_1 \mathrel{Q_i} i$, $Q_j: s^* \mathrel{Q_j} j$, $Q_k: s_1 \mathrel{Q_k} s_2 \mathrel{Q_k} k$, $Q_l = R_l$. Then $TDA(t)(Q) = ((i, s_2), (j, s^*), (k, s_1))$. This makes $s^*$ strictly worse off.

\subsection{Proof of Proposition \ref{prop:st-pf}}
The ``if'' direction is a corollary of \Cref{lem:DA&TDA}. Since the TDA mechanism is outcome-equivalent to the DA mechanism, and the DA mechanism is strategy-proof, we know the TDA mechanism is also strategy-proof.

For the ``only if'' direction, we show that if there exists $i \in I$ such that $R_i \not \in \mathcal{R}^t$, then the TDA mechanism with $t$ is not strategy-proof. 

Since $R_i \not \in \mathcal{R}^t$, there exist schools $s_1$ and $s_2$ such that $t_{s_1} < t_{s_2}$ and $s_2 \mathrel{P_i} s_1$.
Let $\bar S := \{s \in S: s \mathrel{P_i} s_2 \text{ and } t_s \le t_{s_2}\}$ be the set of schools that $i$ strictly prefers to $s_2$ and are in the same or an earlier tier than $s_2$. Then, construct a $(q, \succsim, R_{-i})$ as follows:
\begin{itemize}
    \item For any $s \in \bar S$, $q_s = 0$; $q_{s_2} = q_{s_1} > 0$; other schools' quotas can be arbitrary.\footnote{If we do not allow for zero quotas, we need to assume that $\sum_{s \in \bar S} q_s < |I|$. Otherwise, $i$ could always be admitted to a school in $\bar S$, because we assume all schools rank all students as acceptable.}
    \item $\succsim_{s_2}$ ranks $i$ the highest, other schools' priorities can be arbitrary.
    \item All students $j \in I \backslash \{i\}$ rank $s_1$ and $s_2$ as unacceptable, the preferences for other schools can be arbitrary.
\end{itemize}
Then, $TDA(t)(q, \succsim, R_i, R_{-i})(i) = s_1$.

Suppose $i$ misreports another preference $Q_i$ under which he keeps all preferences in $R_i$ unchanged except for ranking $s_1$ as unacceptable. Then, $TDA(t)(q, \succsim, Q_i, R_{-i})(i) = s_2$, which is preferred by $i$ over $s_1$. Thus, the TDA mechanism with $t$ is not strategy-proof.

\subsection{Proof of Proposition \ref{prop:stable-tier}}
Fix any school choice problem $E$ and tier structure $t$. First, we show that $TDA(t)(R)$ is individually rational. For any $i \in I$, if $s := TDA(t)(R)(i) \in S$, then at some step in the TDA mechanism, $i$ is tentatively accepted by $s$, meaning $i$ and $s$ must be mutually acceptable.

Next, suppose towards a contradiction that there exists a blocking pair $(i, s)$ such that $s \mathrel{P_i} TDA(t)(R)(i)$ and $t_{TDA(t)(R)(i)} \ge t_s$. Since $s \mathrel{P_i} TDA(t)(R)(i)$, $i$ must have proposed to $s$ at some step of the mechanism, and $s$ must have rejected $i$ at some step $k$. Consider possible reasons for this rejection. Given the assumption of $\succsim$, all students are acceptable to $s$. Thus, at step $k$, we must have $q_s = |TDA(t)^{-1}(R)(s)|$ and $j \succ_s i$ for any $j$ that is tentatively accepted by $s$ at this step. Hence, $(i, s)$ is not a blocking pair - a contradiction.

\subsection{Proof of Proposition \ref{prop:weakly}}
Fix $Q_i$ and $Q_i^*$ as defined in \Cref{prop:weakly}, and fix any reporting of others $Q_{-i}$. 
Let us denote the outcome at the end of Round $k$ under the TDA mechanism as $TDA(t)^k$. 

Before the TDA mechanism reaches tier $t_s$, all steps are the same across $(Q_i, Q_{-i})$ and $(Q_i^*, Q_{-i})$, since the two preferences only differ tier $t_s$. Thus, $TDA(t)^{t_s-1}(Q_i, Q_{-i}) = TDA(t)^{t_s-1}(Q_i^*, Q_{-i})$.

If the TDA mechanism reaches tier $t_s$, meaning $TDA(t)^{t_s-1}(Q_i, Q_{-i})(i) = i$, three cases could occur:
\begin{enumerate}
    \item student $i$ is matched to the same school $s^*$ under $Q_i$ and $Q_i^*$, i.e., $TDA(t)^{t_s}(Q_i, Q_{-i})(i) = TDA(t)^{t_s}(Q_i^*, Q_{-i})(i) = s^*$;
    \item $i$ is unmatched in this tier, i.e., $TDA(t)^{t_s}(Q_i, Q_{-i})(i) = TDA(t)^{t_s}(Q_i^*, Q_{-i})(i) = i$;
    \item $i$ is matched to $s_1$ under $Q_i$ and to $s_2$ under $Q_i^*$ (possibly $s_1 = i$ or $s_2 = i$), i.e., $TDA(t)^{t_s}(Q_i, Q_{-i})(i) = s_1$ and $TDA(t)^{t_s}(Q_i^*, Q_{-i})(i) = s_2$.
\end{enumerate}

If Case 1 happens, we know $TDA(t)(Q_i, Q_{-i})(i) = TDA(t)(Q_i^*, Q_{-i})(i) = s^*$. 

If Case 2 happens, we know $TDA(t)^{t_s}(Q_i, Q_{-i})$ is stable with respect to $(\widetilde{Q}_i^*, \widetilde Q_{-i})$, and thus $TDA(t)^{t_s}(Q_i^*, Q_{-i})(j) \mathrel{\widetilde Q_j} TDA(t)^{t_s}(Q_i, Q_{-i})(j)$ for any $j \in I \backslash \{i\}$. Similarly, $TDA(t)^{t_s}(Q_i^*, Q_{-i})$ is stable with respect to $(\widetilde Q_i, \widetilde Q_{-i})$, and thus $TDA(t)^{t_s}(Q_i, Q_{-i})(j) \mathrel{\widetilde Q_j} TDA(t)^{t_s}(Q_i^*, Q_{-i})(j)$. Therefore, $TDA(t)^{t_s}(Q_i^*, Q_{-i})(j) = TDA(t)^{t_s}(Q_i, Q_{-i})(j)$ for any $j \in I \backslash \{i\}$, which implies $TDA(t)^{t_s}(Q_i^*, Q_{-i}) = TDA(t)^{t_s}(Q_i, Q_{-i})$. Since under the TDA mechanism, all steps after tier $t_s$ are the same across $(Q_i, Q_{-i})$ and $(Q_i^*, Q_{-i})$, we know $TDA(t)(Q_i, Q_{-i})(i) = TDA(t)(Q_i^*, Q_{-i})(i)$.

If Case 3 happens, since the DA mechanism is strategy-proof, $s_2 \mathrel{Q_i^*} s_1$. Then, because $s_1 \mathrel{P_i} s_2 \iff s_1 \mathrel{Q^*_i} s_2$, we have $s_2 \mathrel{P_i} s_1$, and thus $TDA(t)(Q^*_i, Q_{-i})(i) \mathrel{P_i} TDA(t)(Q_i, Q_{-i})(i)$.

Thus, we know under any $Q_{-i}$, $TDA(t)(Q_i^*, Q_{-i}) \mathrel{R_i} TDA(t)(Q_i, Q_{-i})$ holds.

\section{Additional Examples} \label{apx:B}
\begin{example}[Moving a school from a later tier to an earlier tier may make it worse off] \label{exp:move}
    Consider three students, $1, 2, 3$, and three schools, $a, b, c$, each with one seat. Student preferences $R$ as well as school priorities $\succsim$ are as follows:
    \begin{multicols}{2}
    \begin{center}
    \begin{tabular}{c|c|c}
        \textbf{$R_1$} & \textbf{$R_2$} & \textbf{$R_3$} \\
        \hline
        $a$ & $b$ & $b$ \\
        $b$ & $a$ & $c$ \\
        $c$ & $c$ & $a$ \\
        $1$ & $2$ & $3$ \\
    \end{tabular}
    \end{center}
   
    \columnbreak
    \begin{center}
    \begin{tabular}{c|c|c}
        \textbf{$\succsim_a$} & \textbf{$\succsim_b$} & \textbf{$\succsim_c$} \\
        \hline
        $3$ & $1$ & $3$ \\
        $1$ & $2$ & $2$ \\
        $2$ & $3$ & $1$ \\
    \end{tabular}
    \end{center}
    \end{multicols}

    Under the tier structure $t = (1, 2, 2)$, $\mathcal{O}^{TDA(t)}(E) = \{((1, a), (2, b), (3, c))\}$. However, under the tier structure $t^\prime = (1, 1, 2)$ where school $b$ is moved to the first tier, $\mathcal{O}^{TDA(t^\prime)}(E) = \{((1, a), (2, b), (3, c)), ((1, a), (2, c), (3, b))\}$. Thus, school $b$ obtains one strictly worse-off equilibrium outcome after being moved to the earlier tier.$\hfill\square$
\end{example}

\begin{example}[Within-tier cycles can be allowed when the priority structure is known] \label{exp:S^*}
    Consider three students, $1, 2, 3$, and three schools, $a, b, c$, each with one seat. School $a$ is in tier 1, and schools $b$ and $c$ are in tier 2. The priority profile $\succsim$ is shown below:
    \begin{center}
    \begin{tabular}{c|c|c}
        \textbf{$\succsim_a$} & \textbf{$\succsim_b$} & \textbf{$\succsim_c$} \\
        \hline
        $2$ & $1$ & $3$ \\
        $1$ & $2$ & $2$ \\
        $3$ & $3$ & $1$ \\
    \end{tabular}
    \end{center}

    Suppose $S^* = \{a\}$. It can be checked that under any preference profile $R \in \mathcal{R}^3$, any Nash equilibrium outcome of the TDA mechanism is weakly better for school $a$ than the corresponding SOSM. Interestingly, the same result holds when the tier structure is changed to $t^\prime = (2, 1, 1)$.\footnote{We verify these results by traversing the set of all possible preference profiles and computing the corresponding Nash equilibrium outcomes. The verification program is available upon request.}$\hfill\square$
\end{example}

\begin{example}[Preference uncertainty] \label{exp:prefun}
    There are three students, $1, 2, 3$, and three schools, $a, b, c$, each of which has one seat. School $a$ belongs to tier 1, and the rest schools are in tier 2. The priority profile $\succsim$ is shown below:
    \begin{center}
    \begin{tabular}{c|c|c}
        \textbf{$\succsim_a$} & \textbf{$\succsim_b$} & \textbf{$\succsim_c$} \\
        \hline
        1 & 1 & 2 \\
        2 & 2 & 1 \\
        3 & 3 & 3 \\
    \end{tabular}
    \end{center}
   
    All students are expected utility maximizers. The types $U_2, U_3$ for students 2 and 3 are known with certainty, and student 1 is one of the two types $U_1^1, U_1^2$ with probability $1/5, 4/5$ respectively. The student types $U$ and corresponding preferences $R$ are as follows:
    \begin{multicols}{2}
        \begin{center}
        \begin{tabular}{c|c|c|c|c|}
             \multicolumn{1}{c}{} & \multicolumn{1}{c}{$U_1^1$}  & \multicolumn{1}{c}{$U_1^2$} & \multicolumn{1}{c}{$U_2$} & \multicolumn{1}{c}{$U_3$}\\\cline{2-5}
              $a$ & 1 & 1 & 2 & 1 \\\cline{2-5}
              $b$ & 3 & 2 & 3 & 3 \\\cline{2-5}
              $c$ & 2 & 3 & 1 & 2 \\\cline{2-5}
              $i$ & 0 & 0 & 0 & 0 \\\cline{2-5}
        \end{tabular}
        \end{center}

        \columnbreak
        \begin{center}
        \begin{tabular}{c|c|c|c}
             \textbf{$R_1^1$} & \textbf{$R_1^2$} & \textbf{$R_2$} & \textbf{$R_3$}\\
            \hline
            $b$ & $c$ & $b$ & $b$\\
            $c$ & $b$ & $a$ & $c$\\
            $a$ & $a$ & $c$ & $a$\\
            1 & 1 & 2 & 3\\
        \end{tabular}
        \end{center}
    \end{multicols}

    Here, we allow each student's strategy space to be a mapping from his type to a preference in $\mathcal{R}$. Thus, we consider the solution concept of Bayesian Nash equilibrium in this example.

    In the dominant-strategy equilibrium of the DA mechanism, where everyone reports his type truthfully, the outcomes are
    $$DA(R^1_1, R_2, R_3) =
    \begin{pmatrix}
    1 & 2 & 3 \\
    b & a & c \\
    \end{pmatrix},
    \text{ }
    DA(R^2_1, R_2, R_3) =
    \begin{pmatrix}
    1 & 2 & 3 \\
    c & b & a \\
    \end{pmatrix},$$ with respect to two realizations of student 1's type.
   
    Considering the preference revelation game induced by the TDA mechanism, one Bayesian Nash equilibrium is shown below, where the reporting behind $i$ is omitted.
    \begin{center}
        \begin{tabular}{c|c|c|c}
            \textbf{$Q_1^1$} & \textbf{$Q_1^2$} & \textbf{$Q_2$} & \textbf{$Q_3$} \\
            \hline
            $b$ & $c$ & $b$ & $b$\\
            1 & 1 & $c$ & $c$\\
             &  & 2 & $a$\\
             &  &  & 3\\
        \end{tabular}
    \end{center}
    The corresponding TDA outcomes are
    $$TDA(t)(Q_1^1, Q_2, Q_3) =
    \begin{pmatrix}
    1 & 2 & 3 \\
    b & c & a \\
    \end{pmatrix},
    \text{ }
    TDA(t)(Q_1^2, Q_2, Q_3) =
    \begin{pmatrix}
    1 & 2 & 3 \\
    c & b & a \\
    \end{pmatrix},$$ with respect to two realizations of student 1's type.

    In fact, the four matchings we derived under the DA and TDA mechanisms are the unique Nash equilibrium outcomes, respectively.

    First, we show that \Cref{thm:nested,thm:implementation} do not extend to this setting. First, note that the unique TDA equilibrium outcome $TDA(t)(Q_1^1, Q_2, Q_3)$ is not stable with respect to $(R^1_1, R_2, R_3)$, as $(2, a)$ forms a blocking pair - a failure of part 1 of \Cref{thm:nested}. Second, the priority profile satisfies within-tier acyclicity - a failure of \Cref{thm:implementation}. Third, the Nash equilibrium outcome of the TDA mechanism $TDA(t)(Q_1^1, Q_2, Q_3)$ is not an equilibrium outcome under the DA mechanism $DA(R_1^1, R_2, R_3)$ - a failure of part 2 of \Cref{thm:nested}.

    Second, we show that the TDA mechanism may not improve the quality of students admitted by top-tier schools as intended.
    Notice that in this example, under type $U^2_1$, $DA(R^2_1, R_2, R_3)$ and $TDA(t)(Q_1^2, Q_2, Q_3)$ are the same. However, under another realization $U^1_1$, the top-tier school $a$ is getting low-priority student 3 instead of its original achievable student 2. Thus, school $a$ is worse off.$\hfill\square$
\end{example}

\begin{example}[Priority uncertainty without alignment in priorities]\label{exp:prioun2}
    There are three students, $1, 2, 3$, and three schools, $a, b, c$, each of which has one seat. School $c$ belongs to tier 1, and the rest schools are in tier 2. Suppose all students are expected utility maximizers, and the following are their types, $U_1, U_2, U_3$, and preferences, $R_1, R_2, R_3$:
    \begin{multicols}{2}
        \begin{center}
        \begin{tabular}{c|c|c|c|}
             \multicolumn{1}{c}{} & \multicolumn{1}{c}{$U_1$}  & \multicolumn{1}{c}{$U_2$} & \multicolumn{1}{c}{$U_3$} \\\cline{2-4}
              $a$ & 3 & 3 & 2  \\\cline{2-4}
              $b$ & 2 & 1 & 3 \\\cline{2-4}
              $c$ & 1 & 2 & 1 \\\cline{2-4}
              $i$ & 0 & 0 & 0 \\\cline{2-4}
        \end{tabular}
        \end{center}

        \columnbreak
        \begin{center}
        \begin{tabular}{c|c|c}
             \textbf{$R_1$} & \textbf{$R_2$} & \textbf{$R_3$}\\
            \hline
            $a$ & $a$ & $b$\\
            $b$ & $c$ & $a$\\
            $c$ & $b$ & $c$\\
            1 & 2 & 3\\
        \end{tabular}
        \end{center}
    \end{multicols}

    There is no uncertainty regarding school $b$ and $c$'s priorities $\succsim_b, \succsim_c$, while there are two types of priority $\succsim_a, \succsim_a^\prime$ for school $a$ with probability $1/5, 4/5$ respectively.
     \begin{center}
        \begin{tabular}{c|c|c|c}
            \textbf{$\succsim_a$} & \textbf{$\succsim_a^\prime$} & \textbf{$\succsim_b$} & \textbf{$\succsim_c$}\\
            \hline
            1 & 2 & 1 & 1\\
            2 & 1 & 2 & 2\\
            3 & 3 & 3 & 3\\
        \end{tabular}
    \end{center}

    As in \Cref{exp:prioun}, in this preference revelation game, we restrict each player's strategy space to $\mathcal{R}$.

    In the dominant-strategy equilibrium of the DA mechanism, where everyone reports truthfully, the outcomes are
    $$DA(q, \succsim, R) = \begin{pmatrix}
    1 & 2 & 3 \\
    a & c & b \\
    \end{pmatrix},
    \text{ }
    DA(q, \succsim^\prime_a, \succsim_{-a}, R) =
    \begin{pmatrix}
    1 & 2 & 3 \\
    b & a & c \\
    \end{pmatrix},$$ with respect to two realizations of the priority profile.

    Considering the preference revelation game induced by the TDA mechanism, one Nash equilibrium is shown below, where the reporting behind $i$ is omitted.
    \begin{center}
        \begin{tabular}{c|c|c}
            \textbf{$Q_1$} & \textbf{$Q_2$} & \textbf{$Q_3$} \\
            \hline
            $a$ & $a$ & $b$\\
            $b$ & $b$ & $a$\\
            1 & 2 & $c$\\
             &  & 3\\
        \end{tabular}
    \end{center}
    The corresponding TDA outcomes are
    $$
    TDA(t)(q, \succsim, Q) = \begin{pmatrix}
    1 & 2 & 3 \\
    a & b & c \\
    \end{pmatrix},
    \text{ }
    TDA(t)(q, \succsim^\prime_a, \succsim_{-a}, Q) =
    \begin{pmatrix}
    1 & 2 & 3 \\
    b & a & c \\
    \end{pmatrix},
    $$ with respect to two realizations of the priority profile.

    Notice that in this example, under the realization $\succsim^\prime_a$, $DA(q, \succsim^\prime_a, \succsim_{-a}, R)$ and $TDA(q, \succsim^\prime_a, \succsim_{-a}, Q)$ are the same. However, under another realization $\succsim_a$, the top-tier school $c$ is getting low-priority student 3 instead of its achievable student 2. Thus, school $c$ is worse off.$\hfill\square$
\end{example}

\newpage
\bibliographystyle{apalike}
\bibliography{main}

\begin{thebibliography}{}

\bibitem[Abdulkadiroğlu et~al., 2005]{abdulkadiroglu-2005}
Abdulkadiroğlu, A., Pathak, P.~A., and Roth, A.~E. (2005).
\newblock {The New York City High School Match}.
\newblock {\em American Economic Review}, 95(2):364–367.

\bibitem[Abdulkadiroğlu et~al., 2009]{abdulkadiroglu-2009}
Abdulkadiroğlu, A., Pathak, P.~A., and Roth, A.~E. (2009).
\newblock {Strategy-Proofness versus Efficiency in Matching with Indifferences: Redesigning the NYC High School Match}.
\newblock {\em American Economic Review}, 99(5):1954–78.

\bibitem[Abdulkadiroğlu and Sönmez, 2003]{abdulkadiroglu-2003}
Abdulkadiroğlu, A. and Sönmez, T. (2003).
\newblock {School Choice: A Mechanism Design Approach}.
\newblock {\em American Economic Review}, 93(3):729--747.

\bibitem[Afacan et~al., 2022]{afacan-2022}
Afacan, M.~O., Evdokimov, P., Hakimov, R., and Turhan, B. (2022).
\newblock {Parallel Markets in School Choice}.
\newblock {\em Games and Economic Behavior}, 133:181--201.

\bibitem[Akahoshi, 2014a]{akahoshi-2014b}
Akahoshi, T. (2014a).
\newblock {A Necessary and Sufficient Condition for Stable Matching Rules to be Strategy-Proof}.
\newblock {\em Social Choice and Welfare}, 43(3):683--702.

\bibitem[Akahoshi, 2014b]{akahoshi-2014}
Akahoshi, T. (2014b).
\newblock {Singleton Core in Many-to-One Matching Problems}.
\newblock {\em Mathematical Social Sciences}, 72:7--13.

\bibitem[Andersson, 2017]{andersson-2017}
Andersson, T. (2017).
\newblock {Matching Practices for Elementary Schools - Sweden}.
\newblock {\em Matching in Practice Country Profile No 24}.

\bibitem[Andersson et~al., 2024]{andersson-2024}
Andersson, T., Dur, U., Ertemel, S., and Kesten, O. (2024).
\newblock {Sequential School Choice with Public and Private Schools}.
\newblock {\em Social Choice and Welfare}, pages 1--46.

\bibitem[Azevedo and Budish, 2018]{azevedo-2018}
Azevedo, E.~M. and Budish, E. (2018).
\newblock {Strategy-Proofness in the Large}.
\newblock {\em Review of Economic Studies}.

\bibitem[Azevedo and Leshno, 2016]{azevedo-2016}
Azevedo, E.~M. and Leshno, J.~D. (2016).
\newblock {A Supply and Demand Framework for Two-Sided Matching Markets}.
\newblock {\em Journal of Political Economy}, 124(5):1235--1268.

\bibitem[Bo and Hakimov, 2022]{bo-2022}
Bo, I. and Hakimov, R. (2022).
\newblock {The Iterative Deferred Acceptance Mechanism}.
\newblock {\em Games and Economic Behavior}, 135:411--433.

\bibitem[Chen and Kesten, 2017]{chen-2017}
Chen, Y. and Kesten, O. (2017).
\newblock {Chinese College Admissions and School Choice Reforms: A Theoretical Analysis}.
\newblock {\em Journal of Political Economy}, 125(1):99--139.

\bibitem[Doğan and Yenmez, 2019]{dogan-2019}
Doğan, B. and Yenmez, M.~B. (2019).
\newblock {Unified Versus Divided Enrollment in School Choice: Improving Student Welfare in Chicago}.
\newblock {\em Games and Economic Behavior}, 118:366--373.

\bibitem[Doğan and Yenmez, 2023]{dogan-2023}
Doğan, B. and Yenmez, M.~B. (2023).
\newblock {When Does an Additional Stage Improve Welfare in Centralized Assignment?}
\newblock {\em Economic Theory}, 76(4):1145--1173.

\bibitem[Dur and Kesten, 2018]{dur-2018}
Dur, U. and Kesten, O. (2018).
\newblock {Sequential versus Simultaneous Assignment Systems and Two Applications}.
\newblock {\em Economic Theory}, 68(2):251--283.

\bibitem[Ekmekçi and Yenmez, 2019]{ekmekci-2019}
Ekmekçi, M. and Yenmez, M.~B. (2019).
\newblock {Common Enrollment in School Choice}.
\newblock {\em Theoretical Economics}, 14(4):1237--1270.

\bibitem[Ergin, 2002]{ergin-2002}
Ergin, H. (2002).
\newblock {Efficient Resource Allocation on the Basis of Priorities}.
\newblock {\em Econometrica}, 70(6):2489--2497.

\bibitem[Ergin and Sönmez, 2006]{ergin-2006}
Ergin, H. and Sönmez, T. (2006).
\newblock {Games of School Choice under the Boston Mechanism}.
\newblock {\em Journal of Public Economics}, 90(1-2):215--237.

\bibitem[Gale and Shapley, 1962]{gale-1962}
Gale, D. and Shapley, L.~S. (1962).
\newblock {College Admissions and the Stability of Marriage}.
\newblock {\em American Mathematical Monthly}, 69(1):9--15.

\bibitem[Haeringer and Iehlé, 2021]{haeringer-2021}
Haeringer, G. and Iehlé, V. (2021).
\newblock {Gradual College Admission}.
\newblock {\em Journal of Economic Theory}, 198:105378.

\bibitem[Haeringer and Klijn, 2009]{haeringer-2009}
Haeringer, G. and Klijn, F. (2009).
\newblock {Constrained School Choice}.
\newblock {\em Journal of Economic Theory}, 144(5):1921--1947.

\bibitem[Hafalir et~al., 2013]{hafalir-2013}
Hafalir, I.~E., Yenmez, M.~B., and Yildirim, M.~A. (2013).
\newblock {Effective Affirmative Action in School Choice}.
\newblock {\em Theoretical Economics}, 8(2):325--363.

\bibitem[Hatakeyama, 2024]{hatakeyama-2024}
Hatakeyama, T. (2024).
\newblock {When is a Sequential School Choice System (Non-)Deficient?}
\newblock {\em Review of Economic Design}, pages 1--12.

\bibitem[Kamada and Kojima, 2015]{kamada-2015}
Kamada, Y. and Kojima, F. (2015).
\newblock {Efficient Matching under Distributional Constraints: Theory and Applications}.
\newblock {\em American Economic Review}, 105(1):67–99.

\bibitem[Kamada and Kojima, 2017]{kamada-2017}
Kamada, Y. and Kojima, F. (2017).
\newblock {Stability Concepts in Matching under Distributional Constraints}.
\newblock {\em Journal of Economic Theory}, 168:107--142.

\bibitem[Kesten, 2010]{kesten-2010}
Kesten, O. (2010).
\newblock {School Choice with Consent}.
\newblock {\em Quarterly Journal of Economics}, 125(3):1297--1348.

\bibitem[Kesten and Kurino, 2019]{kesten-2019}
Kesten, O. and Kurino, M. (2019).
\newblock {Strategy-Proof Improvements upon Deferred Acceptance: A Maximal Domain for Possibility}.
\newblock {\em Games and Economic Behavior}, 117:120--143.

\bibitem[Kumano and Watabe, 2011]{kumano-2011}
Kumano, T. and Watabe, M. (2011).
\newblock {Untruthful Dominant Strategies for the Deferred Acceptance Algorithm}.
\newblock {\em Economics Letters}, 112(2):135--137.

\bibitem[Kumano and Watabe, 2012]{kumano-2012}
Kumano, T. and Watabe, M. (2012).
\newblock {Dominant Strategy Implementation of Stable Rules}.
\newblock {\em Games and Economic Behavior}, 75(1):428--434.

\bibitem[Li and Xu, 2020]{li-2020}
Li, L. and Xu, W. (2020).
\newblock {Analysis of the Influence of College Entrance Examination Admission Batch Reform on Universities and Examinees}.
\newblock {\em Education and Examinations}, 1:48--55.

\bibitem[Manjunath and Turhan, 2016]{manjunath-2016}
Manjunath, V. and Turhan, B. (2016).
\newblock {Two School Systems, One District: What to Do When a Unified Admissions Process is Impossible}.
\newblock {\em Games and Economic Behavior}, 95:25--40.

\bibitem[Pathak and Sönmez, 2008]{pathak-2008}
Pathak, P.~A. and Sönmez, T. (2008).
\newblock {Leveling the Playing Field: Sincere and Sophisticated Players in the Boston Mechanism}.
\newblock {\em American Economic Review}, 98(4):1636--1652.

\bibitem[Pathak and Sönmez, 2013]{pathak-2013}
Pathak, P.~A. and Sönmez, T. (2013).
\newblock {School Admissions Reform in Chicago and England: Comparing Mechanisms by Their Vulnerability to Manipulation}.
\newblock {\em American Economic Review}, 103(1):80–106.

\bibitem[Romero-Medina and Triossi, 2013]{romero-2013}
Romero-Medina, A. and Triossi, M. (2013).
\newblock {Acyclicity and Singleton Cores in Matching Markets}.
\newblock {\em Economics Letters}, 118(1):237--239.

\bibitem[Romero‐Medina, 1998]{romeromedina-1998}
Romero‐Medina, A. (1998).
\newblock {Implementation of Stable Solutions in a Restricted Matching Market}.
\newblock {\em Review of Economic Design}, 3(2):137--147.

\bibitem[Roth, 1985]{roth-1985}
Roth, A.~E. (1985).
\newblock {The College Admissions Problem is not Equivalent to the Marriage Problem}.
\newblock {\em Journal of Economic Theory}, 36(2):277--288.

\bibitem[Roth and Sotomayor, 1990]{roth-1990}
Roth, A.~E. and Sotomayor, M. (1990).
\newblock {\em Two-Sided Matching: A Study in Game-Theoretic Modeling and Analysis}.
\newblock Econometric Society Monographs. Cambridge University Press.

\bibitem[Satterthwaite and Sonnenschein, 1981]{satterthwaite-1981}
Satterthwaite, M.~A. and Sonnenschein, H. (1981).
\newblock {Strategy-Proof Allocation Mechanisms at Differentiable Points}.
\newblock {\em Review of Economic Studies}, 48(4):587.

\bibitem[Sotomayor, 2008]{sotomayor-2008}
Sotomayor, M. (2008).
\newblock {The Stability of the Equilibrium Outcomes in the Admission Games Induced by Stable Matching Rules}.
\newblock {\em International Journal of Game Theory}, 36(3-4):621--640.

\bibitem[Turhan, 2019]{turhan-2019}
Turhan, B. (2019).
\newblock {Welfare and Incentives in Partitioned School Choice Markets}.
\newblock {\em Games and Economic Behavior}, 113:199--208.

\bibitem[Westkamp, 2012]{westkamp-2012}
Westkamp, A. (2012).
\newblock {An Analysis of the German University Admissions System}.
\newblock {\em Economic Theory}, 53(3):561--589.

\end{thebibliography}

\end{document}